\documentclass[11pt, letterpaper]{article}
\usepackage[margin=1in]{geometry}
\usepackage{hyperref, url, booktabs}
\usepackage{amsfonts, nicefrac, microtype}
\usepackage{amsmath,amssymb,amsthm,mathtools,bm, bbm}
\usepackage{marginnote}
\usepackage[normalem]{ulem}
\usepackage{subfigure}
\usepackage{balance}
\usepackage{booktabs} 
\usepackage{dsfont}
\usepackage[ruled]{algorithm2e} 
\usepackage{cleveref}
\usepackage{xkeyval}
\usepackage{xstring}
\usepackage{iftex}
\usepackage{graphicx}
\usepackage[prologue]{xcolor}
\usepackage{geometry}
\usepackage[style=alphabetic, maxbibnames=999,maxalphanames=5,natbib=true]{biblatex}
\newcommand{\R}{\mathbb{R}}
\newcommand{\E}{\mathbb{E}}
\newcommand{\Var}{\textnormal{Var}}
\newcommand{\ind}{\mathds{1}}

\newcommand{\Opt}{\textnormal{OPT}}

\newcommand{\Rev}{\textnormal{Rev}}
\newcommand{\SRev}{\textnormal{SRev}}
\newcommand{\BRev}{\textnormal{BRev}}
\newcommand{\Val}{\textnormal{Val}}
\newcommand{\bt}{\mathbf{t}} 
\newcommand{\bX}{\mathbf{X}} 
\newcommand{\bx}{\mathbf{x}} 
\newcommand{\lt}{t} 
\newcommand{\bss}{\mathbf{s}} 
\newcommand{\lss}{s} 
\newcommand{\mrf}{\mathcal{F}} 
\newcommand{\indep}{\textnormal{ind}}

\def\bigO#1{\operatorname{O}\prn{#1}}

\def\bigOm#1{\operatorname{\Omega}\prn{#1}}

\def\bigTh#1{\operatorname{\Theta}\prn{#1}}

\def\argmax{\operatornamewithlimits{arg\,max}}

\def\E{\mathop{\mathbb{E}}}		

\def\poly{\operatorname{poly}}

\def\1{\mathbbm{1}}     

\def\dif#1{\mathop{d #1}}

\def\floor#1{\left\lfloor #1 \right\rfloor}
\def\ceil#1{\left\lceil #1 \right\rceil}

\def\set#1{\left\{ #1 \right\}}

\def\prn#1{\left( #1 \right)}
\def\brk#1{\left[ #1 \right]}

\def\midd{\:\middle|\:}

\newcommand{\alg}{{\rm ALG}}

\def\cD{\mathcal{D}}

\def\cI{\mathcal{I}}

\def\cP{\mathcal{P}}

\def\f#1#2{\nicefrac{#1}{#2}}

\newtheorem{theorem}{Theorem}[section]
\newtheorem{lemma}[theorem]{Lemma}
\newtheorem{claim}[theorem]{Claim}

\newtheoremstyle{break}
  {\topsep}{\topsep}%
  {\itshape}{}%
  {\bfseries}{}%
  {\newline}{}%
\theoremstyle{break}

\theoremstyle{definition}
\newtheorem{definition}[theorem]{Definition}

\theoremstyle{definition}

\usepackage[textsize=tiny,textwidth=2cm]{todonotes}

\def\prn#1{\left( #1 \right)}
\def\set#1{\left\{ #1 \right\}}
\def\brk#1{\left[ #1 \right]}

\newcommand{\ignore}[1]{}

\newcounter{note}[section]

\definecolor{darkgreen}{RGB}{43, 181, 63}

\allowdisplaybreaks

\title{Improved Mechanisms and Prophet Inequalities for \\ Graphical Dependencies}

\author{
Vasilis Livanos\footnote{(livanos3@illinois.edu) Department of Industrial Engineering, University of Chile.} \and
Kalen Patton\footnote{(kpatton33@gatech.edu) School of Math, Georgia Tech. Supported in part by NSF award CCF-2327010.} \and
Sahil Singla\footnote{(ssingla@gatech.edu) School of Computer Science, Georgia Tech. Supported in part by NSF award CCF-2327010.}
}

\begin{document}


	\maketitle
  \begin{abstract}
Over the past two decades, significant strides have been made in stochastic problems such as revenue-optimal auction design and prophet inequalities, traditionally modeled with $n$ independent random variables to represent the values of $n$ items. However, in many applications, this assumption of independence often diverges from reality. Given the strong impossibility results associated with arbitrary correlations, recent research has pivoted towards exploring these problems under models of mild dependency.

\medskip 

In this work, we study the optimal auction and prophet inequalities problems within the framework of the popular graphical model of Markov Random Fields (MRFs), a choice motivated by its ability to capture complex dependency structures. Specifically, for the problem of selling $n$ items to a single buyer to maximize revenue, we show that $\max\{\SRev, \BRev\}$ is an $O(\Delta)$-approximation to the optimal revenue for subadditive buyers,  
where $\Delta$ is the maximum weighted degree of the underlying MRF. This is a generalization as well as an exponential improvement on the $\exp(O(\Delta))$-approximation results of Cai and Oikonomou~(EC 2021) for additive and unit-demand buyers. We also obtain a similar exponential improvement for the prophet inequality problem, which is asymptotically optimal as we show a matching upper bound.
\end{abstract}
\newpage





\section{Introduction}\label{sec:intro}

In several stochastic optimization problems arising in economics, \emph{item-independence} assumptions are frequently used to avoid strong negative results that exist for general distributions. While such independence assumptions are useful for theoretical guarantees, they are not necessarily realistic in practice. In auction design, for instance, it is likely that similar goods have values which are positively correlated. This inspires a different research direction: stochastic problems with \emph{bounded} correlation strength. Such investigation brings theoretical results closer to practice while avoiding the hardness of arbitrary distributions. In this paper, we study how mild dependencies affect the hardness of two classic stochastic problems: revenue maximizing auction design and the prophet inequality problem.


\medskip
\noindent\textbf{Auction Design. } Consider the problem of selling a set $[n]$ of $n$ items to a single buyer, whose valuation function $v : 2^{[n]} \to \R_+$ is private but drawn from a known distribution. Our goal is to design an auction/mechanism that maximizes the expected revenue. In the standard setting, we assume that the item values $v(\{i\})$ are independent for each $i$, and that $v$ is \emph{additive} (i.e., $v(S) = \sum_{i \in S} v(\{i\})$ for all $S$) or \emph{unit-demand} (i.e., $v(S) = \max_{i \in S} v(\{i\})$ for all $S$). However, even with this independence assumption, the revenue-optimal auction is complex (i.e., non-deterministic \cite{Thanassoulis2004HagglingOS, MANELLI20061}, non-monotone \cite{nonmonotonicity}, and intractable to compute \cite{daskalakis2014complexity}). Due to this, one line of work on this problem has focused on finding ``simple'' mechanisms that are approximately optimal. First, for a unit-demand buyer, a simple mechanism that achieves a 4-approximation of optimal revenue was found by \citet{unit-demand-2010} building on the work of \citet{unit-demand-2007}. The case of an additive buyer was later resolved by the landmark work of \citet{matt-optimal-additive}, who show that the greater of the revenue from selling all items separately ($\SRev$) and that of selling all items in a grand bundle ($\BRev$) is a $6$-approximation of the optimal revenue. Following this result, the maximum of $\SRev$ and $\BRev$ was also shown to be a $\bigO{1}$ approximation for the more general class of \emph{subadditive} buyer valuations (i.e., $v(S \cup T) \leq v(S) + v(T)$) \cite{rubinstein-weinberg-subadditive}.

However, all the above results still require item values to be independent. If we allow the buyer's value on the items to be arbitrarily correlated, then getting \emph{any} approximation on the optimal revenue becomes impossible for simple mechanisms \cite{Hart2019}. Nevertheless, the dependent setting is not entirely hopeless, as this hardness only applies for arbitrarily strong correlations. This leaves open the possibility for simple mechanisms to perform well when item values only have mild dependencies.

\medskip
\noindent\textbf{Prophet Inequality.}
Prophet inequalities are another important class of problems where item-independence is traditionally assumed. In this problem, we are presented with $n$ items of unknown values drawn from known distributions. We are allowed to sequentially probe the values of the items in a given order, but after probing each item, we must immediately decide to either take the item and end the game, or to discard the item and proceed to the next. Our goal is to maximize the expected value of the item we ultimately take. 

If we assume item values to be drawn independently, then simple threshold algorithms can achieve $\frac{1}{2}$ of the expected maximum item value \cite{SamuelCahn84,KleinbergW12}. However, if we allow the item values to have arbitrary correlations, then there exist distributions where no algorithm can obtain better than a $\frac{1}{n}$ fraction of the maximum \cite{hill1992survey,lin-correlations-pi}. Again, this motivates us to examine the problem with mild correlations to interpolate between the extremes.

\medskip
\noindent\textbf{Markov Random Fields.} To study stochastic problems with correlations, we use \emph{Markov Random Fields (MRFs)} to model the joint distribution over item values. 
Such models have been successfully employed in various fields, ranging from statistical physics to computer vision, to model high-dimensional distributions  (see, e.g., books \citet{KF-Book09,JN-Book07,Pearl-Book09,Edwards-Book12}).  They represent a collection of dependent random variables as vertices in a hypergraph, where edge weights indicate dependencies among variables.  
For us, MRFs are an appealing model of correlation because (1) they can capture arbitrary distributions, and (2) they allow us to parameterize the strength of correlations through graph statistics like the \emph{maximum weighted  degree} $\Delta$. For example, when $\Delta = 0$, one recovers the independent-items setting, and when $\Delta \to +\infty$, they can capture arbitrary distributions. By examining MRF correlations with $0 < \Delta < \infty$, we can smoothly interpolate between these two extremes.

The study of stochastic problems with MRF dependencies was initiated by \citet{cai-oikonomou-mrf}. They gave an $e^{\bigO{\Delta}}$-approximate algorithm for the prophet inequality problem, and a simple $e^{\bigO{\Delta}}$-approximation for the revenue maximization problem with a single additive or unit-demand buyer. Since $\poly(\Delta)$ lower bounds on the approximation ratio can be shown for both the problems, this leaves open  whether the exponential dependence on $\Delta$ is necessary.

\subsection{Our Results and Techniques}

\noindent \textbf{Auction Design.} 
In the problem of selling $n$ items to a single additive buyer with MRF valuations, we seek to determine the best approximation of the optimal revenue achievable with simple mechanisms. The mechanisms we focus on are separate item pricing ($\SRev$) and grand-bundle pricing schemes ($\BRev$),  which together achieve constant approximation in many of the previously mentioned item-independent settings. Additionally, \citet{cai-oikonomou-mrf} showed that under MRF dependencies, $\max\{\SRev, \BRev\}$ achieves an $e^{\bigO{\Delta}}$-approximation of the optimal revenue.

We improve this factor exponentially to $\bigO{\Delta}$, leaving only a polynomial gap to the lower bound of $\bigOm{\Delta^{1/7}}$ from \citet{cai-oikonomou-mrf}.

\begin{theorem}\label{thm:additive-rev}
For a single additive buyer with valuations given by an MRF with maximum weighted degree $\Delta$, the revenue of the optimal auction is at most $(44 \Delta + 12) \cdot \SRev + 70 (\Delta + 1) \cdot \BRev$.
\end{theorem}

Our techniques also yield results for buyers with valuations beyond additive. For a single unit-demand buyer, we find that $\SRev$ alone is an $\bigO{\Delta}$ approximation. Again, this improves on the $e^{\bigO{\Delta}}$ bound from \citet{cai-oikonomou-mrf}.

\begin{theorem}\label{thm:ud-rev}
For a single unit-demand buyer with valuations given by an MRF with maximum weighted degree $\Delta$, the revenue of the optimal auction is at most $\prn{44 \Delta + 14} \cdot \SRev$.
\end{theorem}

Finally, we generalize \Cref{thm:additive-rev} to the setting of a single subadditive buyer, i.e.,  a buyer whose valuation function $v$ satisfies $v(S \cup T) \leq v(S) + v(T) $ for all subsets $S,T \subseteq [n]$. This setting captures the XOS valuations studied by \citet{cai-oikonomou-mrf}, who show that $\max\{\SRev, \BRev\}$ achieves $\prn{e^{\bigO{\Delta}} + \frac{1}{\sqrt{n \gamma}}}$-approximation. Here, $\gamma$ is a factor which depends on the Glauber dynamics of the MRF\footnote{Specifically, $\gamma$ is the \emph{spectral gap} of the Glauber dynamics -- for more information see \cite{cai-oikonomou-mrf}}. Our results eliminate the dependence on $\gamma$ and $n$, improve the $\Delta$ dependency to $\bigO{\Delta}$, and extend beyond XOS to the larger class of subadditive valuations. We note that this approximation ratio mirrors the $\bigO{1}$ factor achievable in the independent-item setting for a subadditive buyer \cite{rubinstein-weinberg-subadditive}.

\begin{theorem}\label{thm:subadd-rev}
For a single subadditive buyer with valuations given by an MRF with maximum weighted degree $\Delta$, the revenue of the optimal auction is at most $\prn{348 \Delta + 110} \cdot \BRev + 10 \cdot \SRev$.
\end{theorem}

\medskip
\noindent \textbf{High-level Technique.} 
To get the above results, we adapt the  approach of \citet{matt-optimal-additive} for the independent-item setting. This involves decomposing the item set into a ``core'' set containing low-value items and a ``tail'' set containing high-value items. To apply this method for dependent items, we need two new key ingredients. 

First, we develop a new ``approximate marginal mechanism'' lemma (\Cref{lem:gen-splitting}) for correlated items. This allows us to partition the item set $[n]$ into a core set $C$ and tail set $T$, and upper bound the optimal revenue of selling items in $[n]$ by the total value of $C$ and the optimal revenue of selling $T$. The crucial difference between \Cref{lem:gen-splitting} and a similar looking lemma in \citet{rubinstein-weinberg-subadditive} is that our bound uses the unconditional revenue from selling $T$, but the latter lemma would instead use the (larger) revenue from selling $T$ conditional on the valuation function on $C$. For MRF valuations, this conditioning may lose a factor of $e^{\bigOm{\Delta}}$, so avoiding it is necessary to get bounds which are sub-exponential in $\Delta$ .

Second, to handle the MRF correlations in our core and tail bounds, we employ an exponential bucketing technique. The key idea is that, although item values are dependent, many distributional statistics on one item (e.g., expectation, variance, CDF, etc.) can only change by a factor of $e^{\bigO{\Delta}}$ after conditioning on any event regarding the other items. Intuitively, this allow us to guess the ``scale'' of the problem up to a constant factor by guessing which of the $\log(e^{\bigO{\Delta}}) = \bigO{\Delta}$ many buckets we fall into. This bucketing is where we lose the $\bigO{\Delta}$ factor in our approximation ratios.

\medskip
\noindent \textbf{Prophet Inequality.} 
For the prophet inequality problem with MRF dependencies, the prior work of \citet{cai-oikonomou-mrf} gives an $e^{\bigO{\Delta}}$-competitive algorithm. We improve this by designing a $\bigO{\Delta}$-competitive algorithm, and also provide an $\bigOm{\Delta}$ lower bound.

\begin{theorem}\label{thm:prophet}
For any MRF with maximum weighted degree $\Delta$, there exists a $(20 \Delta + 15)$-competitive prophet inequality algorithm. Moreover, there exists an MRF instance for which no online algorithm is better than $\frac{\Delta + 1}{2}$-competitive.
\end{theorem}

The proof of this result uses similar bucketing techniques as in the above auction design problems. We use a single threshold algorithm where our threshold is randomly chosen from one of $\bigTh{\Delta}$ geometrically increasing levels (buckets). As long as the largest item value  falls in to one of the ``buckets'' between these levels, we obtain a constant fraction of the maximum value with probability $\frac{1}{\bigTh{\Delta}}$. In addition, we can show that if the maximum is smaller than the smallest level or larger than the largest level, then we only lose a constant factor of the expected maximum; in the former case because the contribution to the expected maximum is very small and in the latter case because the largest level is large enough that getting more than one realization above the largest level, and thus missing out on the maximum, has a very small probability of occurring.

We also show in \Cref{apx:ocrs} that the popular technique of designing prophet inequalities for independent distributions via ``online contention resolution schemes'' (OCRS) cannot be used to design $\bigO{\Delta}$-competitive algorithms for MRF dependencies. Specifically, in \Cref{thm:ocrs} we show that any approach via OCRS loses an $e^{\bigOm{\Delta}}$ factor in the competitive ratio.

\begin{theorem}\label{thm:ocrs}
For every MRF, there exists a $\prn{\frac{1}{1+e^{4 \Delta}}}$-selectable OCRS. Furthermore, for each $\Delta > 0$, there exists an MRF for which there is no $\alpha$-selectable OCRS for $\alpha \geq 4 e^{-\Delta}$.
\end{theorem}

\subsection{Further Related Work}

There is a long line of work on both revenue maximizing auction design and prophet inequalities for independent distributions. We refer the readers to the books and surveys \cite{roughgarden2017twenty,HartlineBook,correa2019recent,lucier2017economic}. Below we discuss works  that study these problems under correlations.

\medskip
\noindent\textbf{Linear Correlations.} A natural model for dependencies is that of \emph{linear correlations}, where each item value is a linear combination over a common set of independent variables. This model has received attention for both auction design \cite{unit-demand-2010, base-value-additive} and the prophet inequality problem \cite{lin-correlations-pi}. In particular, for the special case of the \emph{base-value} model, in which $v(\{i\}) = X_0 + X_i$ for independent variables $X_0, X_1, \dots, X_n$, constant approximations are known for both problems. However, it is worth noting that in contrast with MRFs, linear correlations are unable to capture arbitrary joint distributions.

\medskip
\noindent\textbf{Pairwise Independence.}
A second model for relaxing  item-independence  is to only assume pairwise independence between item values. Under this weaker notion of independence, \citet{pairwise-indep} showed that constant factor approximations still exist for single-item prophet inequalities and certain revenue maximization problems. Recently, \citet{dughmi2023limitations} extended some of these results to multiple-item settings. They also observed a gap between prophet inequalities and OCRS for pairwise-independent distributions. However, similar to linear correlations, pairwise independence is unable to capture arbitrary joint distributions.

\medskip
\noindent\textbf{Arbitrary Correlations.}
While the problem of selling $n$ items to one buyer with arbitrary correlations suffers from the aforementioned impossibility results, the related problem of selling one item to $n$ buyers is more tractable. Even under arbitrary correlations, \citet{ronen-lookahead} showed that a  ``lookahead'' auction obtains a constant fraction of the optimal revenue.

\section{Model and a New Marginal Mechanism}\label{sec:preliminaries}

In this section, we formally define Markov Random Fields and the Optimal Auction Design problem, and then prove a new approximate marginal mechanism for correlated distributions, which will play a central role in all our revenue maximization results.

\subsection{Markov Random Field Model for Buyer Valuations} 
\label{sec:MRF}

We begin with a formal definition of Markov Random Fields.

\begin{definition}\label{def:mrf}
A \emph{Markov Random Field (MRF)} consists of a tuple 
\[
\mrf = (\{\Omega_i\}_{i \in [n]},\; E,\; \{\psi_i\}_{i \in [n]},\; \{\psi_e\}_{e \in E}),
\]
where
\begin{itemize}
    \item $\Omega = \Omega_1 \times \dots \times \Omega_n$ is the support of the distribution.
    \item $E \subseteq 2^{[n]}$ is the edge set of a hypergraph.
    \item $\psi_i : \Omega_i \to \R$ is a potential function on coordinate $i$ for each $i \in [n]$.
    \item $\psi_e : \prod_{i \in e} \Omega_i \to \R$ is a potential function on hyperedge $e$ for each $e \in E$.
\end{itemize}
A sample from the MRF is a random vector $\bt = (t_1, \dots, t_n)$ supported on $\Omega = \Omega_1 \times \dots \times \Omega_n$ with a probability function given by
\[
\Pr\brk{\bt = \bss} \propto \exp \prn{\sum_{i \in [n]} \psi_i(\lss_i) + \sum_{e \in E} \psi_e (\bss_e) }.
\]
\end{definition}

In this definition, the $\psi_e$ functions impose dependencies between the $t_i$ values. If all $\psi_e$ functions are identically $0$, then $\bt$ has a product distribution. Thus, we can bound the strength of dependencies in $\bt$ by the magnitude of the contributions from $\psi_e$ terms. A standard way to quantify this dependency is the maximum weighted degree of the MRF.

\begin{definition}
The \emph{maximum weighted degree} $\Delta(\mrf)$ of an MRF is given by
\[
\Delta(\mrf) = \max_{i \in [n]} \max_{\bss \in \Omega} \Big|\sum_{e \ni i} \psi_e(\bss_e) \Big|.
\]
When the MRF is clear from context, we will simply write $\Delta$ instead of $\Delta(\mrf)$.
\end{definition}

A crucial property of the parameter $\Delta$ is that it gives us a way to bound the change in the probability of some event involving $\lt_i$ after conditioning on some event over the remaining variables $\bt_{-i}$. This is formally given by the following lemma from \citet{cai-oikonomou-mrf}.

\begin{lemma}[\cite{cai-oikonomou-mrf}, Lemma 2]\label{lem:mrf-conditioning}
    For an MRF sample $\bt$ and index $i \in [n]$, we have for any $E_i \subseteq \Omega_i$ and $E_{-i} \subseteq \Omega_{-i}$,
    \[
    e^{-4\Delta}
    \leq \frac{\Pr(\lt_i \in E_i \wedge \bt_{-i} \in E_{-i})}{\Pr(\lt_i \in E_i) \cdot \Pr(\bt_{-i} \in E_{-i})}
    \leq e^{4\Delta}.
    \]
\end{lemma}

Next, given an MRF, we describe how it naturally implies a model for buyer valuations.

\begin{definition}[Buyer Valuation Distribution $D$]\label{def:mrf-valuations}
Given an MRF and a monotone set function $g : 2^{\Omega_1 \cup \dots \cup \Omega_n} \to \R_+$ (we assume  $\Omega_i$ to be disjoint), the buyer valuation  $v : 2^{[n]} \to \R_+$ 
is obtained by first sampling a (private) type vector $\bt \in \Omega_1 \times \dots \times  \Omega_n$ from the MRF and then  for a set $S$ of items,
\[
v(S) \coloneqq g(\set{t_i : i \in S}).
\]
We let $D$ denote the distribution of buyer valuation $v$.
\end{definition}

Thus, $v$ is drawn implicitly from a distribution over monotone valuation functions. We  use the notation $v(i) = v(\set{i})$ for $i \in [n]$.

We will focus on three special cases of  buyer's valuation: 
\begin{itemize}
    \item \emph{Additive}: if $v(S) = \sum_{i \in S} v(i)$ for all sets $S \subseteq [n]$. 
    \item \emph{Unit-demand}:  if $v(S) = \max_{i \in S} v(i)$ for all sets $S \subseteq [n]$.
    \item \emph{Subadditive}: if $v(A \cup B) \leq v(A) + v(B)$ for all sets $A,B \subseteq [n]$.
\end{itemize} 

\subsection{Auction Design Model}

We study the problem of designing a mechanism to sell a set $[n]$ of $n$ items to a single buyer. The buyer's valuation   is given by $v$, where $v(S)$ denotes the value for item set $S \subseteq [n]$, that is sampled from a distribution $D$ over monotone valuation functions  known to the seller. 

From the classic Taxation Principle~\cite{Hammond79,Guesnerie81}, any single buyer mechanism can be represented by a menu $M = \{(\Gamma_k, p_k) : k \in [|M|]\}$ of options.  Each menu option $(\Gamma_k, p_k)$ consists of a distribution $\Gamma_k$ over subsets of $[n]$ offered at a price $p_k \in \R_+$. If the buyer chooses menu option $k$, they pay price $p_k$ and receive an item set $S$ sampled from $\Gamma_k$. We assume that, given menu $M$, a buyer with valuation $v$ always chooses the option $k^*$ that maximizes their expect utility, i.e.,
\[
k^* = \argmax_{k \in [|M|]}\set{\E_{S \sim \Gamma_{k}}[v(S)] - p_k}.
\]
We will therefore directly refer to a mechanism $M$ as a menu/list of such pairs (which implicitly includes the pair $(\emptyset, 0)$).
The goal of the seller is to design a menu which maximizes the expected revenue $\E_{v \sim D}\brk{p_{k^*}}$. 

We use the notation $\Rev(D)$ to denote the maximum possible revenue from any menu, where the buyers valuation function is drawn from distribution $D$. It will also be convenient for us to use the notation $\Rev_i(D) := \max_{p > 0} \set{p \cdot \Pr\brk{v(\{i\}) \geq p}}$ to denote the optimal revenue of selling the single item $i$, and to use $\Val(D) = \E_{v \sim D}[v([n])]$ to denote the expected buyer's valuation on all items.

Additionally, we define $\SRev(D)$ to be the optimal revenue from selling each item individually. In other words, $\SRev(D)$ is the optimal revenue for a menu of the form $\{(S, \sum_{i \in S} p_i) : S \subseteq [n]\}$ for some collection of prices $p_1, \dots, p_n$. Notice that if $v$ is additive, then $\SRev(D) = \sum_{i \in [n]} \Rev_i(D)$, but this is not necessarily true otherwise. We also define $\BRev(D) = \max_{p > 0} \set{p \cdot \Pr\brk{v([n]) \geq p}}$ to be the maximum revenue obtainable from selling the grand bundle.

\subsection{Approximate Marginal Mechanism for Subadditive Valuations}\label{sec:subadditive-splitting-lemma}


The following approximate marginal mechanism lemma for correlated distributions will be very useful in the analysis of our mechanisms. It provides an upper bound to the optimal revenue via the revenue and welfare of two subsets of items. We will apply it  after partitioning the items into a ``core'' and a ``tail''. 

\begin{lemma}\label{lem:gen-splitting}
Suppose we have two disjoint sets of items $A, B$, and we are selling $A \cup B$ to a single buyer with a random monotone subadditive valuation $v : 2^{A \cup B} \to \R_+$ drawn from a known distribution $D$. Then, the optimal revenue is
\[
\Rev(D) \leq 2 \prn{\Val\big(D^A\big) + \Rev \big(D^B\big)},
\]
where $D^S$ denotes the distribution of $v$ restricted to a set $S$ of items.
\end{lemma}
\begin{proof}
    Consider an optimal menu $M = \set{\prn{\Gamma_k, p_k} : k \in [|M|]}$ for $A \cup B$. We construct a menu 
    \[
    M_B := \set{\prn{\Gamma_k^B, \frac{p_k}{2}} : k \in [|M|]},
    \]
    for $B$ by restricting each allocation to $B$, and discounting prices by a factor of $2$. Here, $\Gamma_k^B$ is the distribution of $S \cap B$, where $S \sim \Gamma_k$.

    For a fixed realization of $v$, let $k^* = \argmax_k \set{\E_{S \sim \Gamma_k} [v(S)] - p_k}$ be the menu option the buyer would choose from $M$ and $k^*_B = \argmax_k \prn{\E_{S \sim \Gamma_k} [v(S \cap B)] - \frac{p_k}{2}}$ be the menu option the buyer would choose from $M_B$, respectively. Let $S^* \sim \Gamma_{k^*}$ and $S_B^* \sim \Gamma_{k_B^*}$. Then, from the optimality of $k^*$ and $k^*_B$ and monotonicity of $v$, we have
    \begin{align*}
        \E_{S^*} [v(S^*)] - p_{k^*} &\geq \E_{S^*_B} [v(S^*_B)] - p_{k^*_B}\\
        &\geq \prn{\E_{S^*_B} [v(S^*_B \cap B)] - \frac{p_{k^*_B}}{2} }  - \frac{p_{k^*_B}}{2} ~\geq~ \prn{\E_{S^*} [v(S^* \cap B)] - \frac{p_{k^*}}{2} } - \frac{p_{k^*_B}}{2}.
    \end{align*}
    Rearranging this inequality gives 
    \[
    \frac{p_{k^*}}{2} \leq \E_{S^*}\brk{v(S^*) - v(S^* \cap B)} + \frac{p_{k^*_B}}{2}.
    \]
    By the subadditivity of $v$ and $S^* \subseteq A \cup B$, we know that $v(S^*) - v(S^* \cap B) \leq v(S^* \cap A) \leq v(A)$. Hence, taking expectation over $v$, and using $\Val\prn{D^A} = \E[v(A)]$, $\Rev(D) = \E_v [p_{k^*}]$, and $\Rev(D^B) \geq \E_v\brk{\frac{p_{k^*_B}}{2}}$, we have 
    \[
    \frac{1}{2} \: \: \Rev(D) \leq \Val\big(D^A\big) + \Rev \big(D^B \big).   \qedhere
    \]
\end{proof}

Now that we have shown the approximate marginal mechanism lemma, we can prove the following crude approximation bound for $\Rev(D)$, which will be very useful in all our results. Given Lemma~\ref{lem:gen-splitting}, its proof is not that difficult and for this reason it has been moved to Appendix~\ref{apx:rev-n-srev-proof}.

\begin{lemma}\label{lem:gen-rev-n-srev}
For a single subadditive buyer with MRF valuations $v(S) = g(\{t_i : i \in S\})$, we have
\[
\Rev(D) \leq 2(\rho + 1) e^{4\Delta} \cdot \sum_{i \in [n]} \Rev_i(D) \qquad \text{where} \qquad \rho = \max_{\substack{j \in [n] \\ \bss \in \Omega}}\frac{g(\{\lss_i : i \neq j\})}{\max_{i \neq j} g(\{\lss_i\})}.
\]
\end{lemma}
Notice that for subadditive buyer valuations  $\rho \leq n-1$ and for unit-demand buyers  $\rho = 1$.

\section{Additive Buyer Mechanisms}\label{sec:additive}
For an additive buyer with MRF valuations, we assume the buyer samples a private type vector $\bt = (\lt_1, \dots, \lt_n)$ from an MRF. The buyer's valuation  $v : 2^{[n]} \to \R_+$ is then given by
\[
v(S) \coloneqq \sum_{i \in S} g(\lt_i),
\]
where we employ the notation $g(t_i) = g(\{t_i\})$.

\subsection{\texorpdfstring{$\SRev$}{SRev} is \texorpdfstring{$O(\log n + \Delta)$}{O(logn + Δ)}-approximate}

We will show that $\SRev$ already gets a decent approximation to the optimal revenue.

\begin{theorem}\label{thm:additive-srev-log-approx}
For selling $n$ items to a single additive buyer with MRF valuations drawn from distribution $D$, we have
\[
\Rev(D) \leq \prn{12 + 16 \Delta + 2 \ln{n}} \cdot \SRev(D).
\]
\end{theorem}





\begin{proof}[Proof of \Cref{thm:additive-srev-log-approx}]
For a given type vector $\bt = (\lt_1, \dots, \lt_n)$, we first partition the set of items into the \emph{core} $C$ and the \emph{tail} $T$, where   the tail $T \subseteq [n]$ is 
\[
T \coloneqq \set{i \in [n] : g(t_i) \geq e^{8\Delta} \SRev(D) }
\]
and the core is $C \coloneqq [n] \setminus T$. Intuitively, the tail represents the set of items that take exceptionally large values compared to $\SRev$.

Additionally, since we would like to be able to condition on the tail set, we use $D_A$ for $A \subseteq [n]$ to denote the conditional distribution of $v$ on the event $T = A$.
 We also define the core and tail components, $v^C$ and $v^T$, of the the valuation $v$ as
\[
v^T(S) \coloneqq v(S \cap T) = \sum_{i \in S} g(\lt_i) \cdot \ind_{i \in T} \qquad \text{and} \qquad 
v^C(S) \coloneqq v(S \cap C) = \sum_{i \in S} g(\lt_i) \cdot \ind_{i \in C}.
\]
We let $D^T$ denote the distribution of $v^T$, and $D^T_A$ to denote its distribution conditional on $T = A$. We similarly define $D^C$ and $D^C_A$ as the distribution of $v^C$ and the distribution of $v^C$ conditional on $T = A$, respectively.

Notice that from the approximate marginal mechanism \Cref{lem:gen-splitting}, we can bound  
\begin{align}
\Rev(D) &\leq \sum_{A \subseteq [n]} \Pr\brk{T = A} \cdot \Rev(D_A) \notag \\
&\leq 2 \sum_{A \subseteq [n]} \Pr\brk{T = A} \cdot \prn{\Val\prn{D^C_A} + \Rev\prn{D^T_A}} \notag  \\
&= 2 \, \Val\prn{D^C} + 2\sum_{A \subseteq [n]} \Pr\brk{T = A} \cdot \Rev\prn{D^T_A}.      \label{eq:addCoreTailDecomp}
\end{align}
We seek to bound both the core and tail contributions in terms of $\SRev(D)$. 

\medskip
\noindent \textbf{Core Contribution.} First, we bound the core contribution.
\begin{claim}\label{claim:coreAdditive} The core contribution is
    \[
\Val\prn{D^C}  \leq \prn{1 + 8\Delta +\ln(n)}\cdot \SRev(D).
\]
\end{claim}

\begin{proof}[Proof of \Cref{claim:coreAdditive}]
To bound the core contribution $\Val\prn{D^C}$, notice that  
\begin{align*}
\Val\prn{D^C} = \sum_{i \in [n]} \E\brk{g(\lt_i) \cdot \ind_{i \in C}}.
\end{align*}
We will bound each term of this sum separately. Let $r = \SRev(D)$ and $r_i = \Rev_i(D)$ for simplicity. Notice that 
\begin{align*}
\E\brk{g(\lt_i) \cdot \ind_{i \in C}} &= \int_0^\infty \Pr\brk{g(\lt_i) \cdot \ind_{i \in C} \geq \tau} \dif \tau  ~\leq~ \int_0^{re^{8\Delta}} \Pr\brk{g(\lt_i) \geq \tau} \dif \tau.
\end{align*}
Notice that since $r_i = \sup_{\tau \geq 0} \tau \cdot \Pr\brk{g(\lt_i) \geq \tau}$, we have  $\Pr\brk{g(\lt_i) \geq \tau} \leq \min\set{1, \frac{r_i}{\tau}}$ for all $\tau > 0$. Thus, 
\begin{align*}
\E\brk{g(\lt_i) \cdot \ind_{i \in C}} &\leq r_i + \int_{r_i}^{re^{8\Delta}} \frac{r_i}{\tau} \dif \tau = r_i \prn{1 + \ln\prn{re^{8\Delta}} - \ln(r_i) } = r_i (1 + 8\Delta + \ln(r / r_i)).
\end{align*}
Summing over all $i$ gives
\begin{align*} 
\Val\prn{D^C} &\leq \sum_{i \in [n]} r_i (1 + 8\Delta + \ln(r / r_i)) = r  \prn{ 1 + 8\Delta - \sum_{i \in [n]} \frac{r_i}{r} \ln(r_i/r) } \leq r \prn{1 + 8\Delta +\ln(n)},
\end{align*}
where the last inequality follows by noticing that $- \sum_{i \in [n]} \frac{r_i}{r} \ln(r_i/r)$ corresponds to the entropy of the distribution that picks each item $i$ independently with probability $\frac{r_i}{r}$, which is maximized when all terms are equal, and thus $- \sum_{i \in [n]} \frac{r_i}{r} \ln(r_i/r) \leq \ln(n)$.
\end{proof}

\medskip
\noindent \textbf{Tail Contribution.} 
Next, we bound the contribution of the tail. Because the tail is small in expectation, most of the revenue should be generated when only one item appears in the tail. This contribution is easily bounded by $\SRev$. When multiple items appear, we can use the coarse bound on $\Rev$ in terms of $\SRev$ given by \Cref{lem:gen-rev-n-srev} and exploit that the expected size $\E[|T|]$ is exponentially small in $\Delta$. Formally, we prove the following claim.

\begin{claim}\label{lem:rev-tail} The tail contribution is
\[
\sum_{A \subseteq [n]} \Pr\brk{T = A} \cdot \Rev\prn{D^T_A} \leq 5 \cdot \SRev\prn{D^T}.
\]
\end{claim}
\begin{proof}
    First, we split the L.H.S. into  two cases:
    \begin{align*}
        \sum_{A \subseteq [n]} \Pr\brk{T = A} \Rev\prn{D^T_A}
        ~\leq~ &\sum_{\substack{A \subseteq [n] \\ |A| \geq 2}} \Pr\brk{T = A} \Rev\prn{D^T_A}  +  \sum_{i \in [n]} \Pr\brk{T = \{i\}} \Rev\prn{D^T_{\{i\}}}.
    \end{align*}
    Clearly, we have $\Pr\brk{T = \{i\}} \cdot \Rev\prn{D^T_{\{i\}}} \leq \Rev_i\prn{D^T}$, so the latter summation is bounded by $\sum_i \Rev\prn{D^T_i} = \SRev\prn{D^T}$.
    Thus, we just need to focus on the contribution of the former summation. From \Cref{lem:gen-rev-n-srev}, we have
    \begin{align*}
        \Rev\prn{D^T_A} ~\leq~ 2|A| e^{4\Delta} \SRev\prn{D^T_A}  
        ~=~2|A| e^{4\Delta} \sum_{i \in A} \Rev_i\prn{D^T_A} 
        ~\leq~ 2|A| e^{8\Delta} \sum_{i \in A} \frac{\Rev_i\prn{D^T}}{\Pr\brk{i \in T}}.
    \end{align*}
    Here, the last step comes from \Cref{lem:mrf-conditioning} as follows. Suppose $i \in A$, and let $p$ be the optimal price for single item mechanism $\Rev_i(D^T_A)$. We have
    
    \begin{align*}
        \Rev_i(D^T) &\geq p \cdot \Pr\brk{g(t_i) \geq p} \\
        &\geq e^{-4\Delta}  \cdot p \cdot \Pr\brk{g(t_i) \geq p \mid T \setminus \{i\} = A \setminus \{i\}} \\
        &= e^{-4\Delta}  \cdot p \cdot \Pr\brk{g(t_i) \geq p \mid T = A} \cdot \Pr\brk{T = A \mid T \setminus \{i\} = A \setminus \{i\}}\\ 
        &= e^{-4\Delta} 
        \cdot \Rev_i(D^T_A) \cdot 
        \frac{\Pr\brk{T = A}}{\Pr\brk{T \setminus \{i\} = A \setminus \{i\}}}.
    \end{align*}
    Now, substituting this bound back into our sum over $|A| \geq 2$, we have
     \begin{align*}
        \sum_{\substack{A \subseteq [n] \\ |A| \geq 2}} \Pr\brk{T = A} \cdot \Rev\prn{D^T_A}
        &\leq 2e^{4\Delta} \sum_{\substack{A \subseteq [n] \\ |A| \geq 2}} \sum_{i \in A} |A| \cdot \Pr\brk{T = A}\cdot  \Rev_i\prn{D^T_A}\\
        &\leq 2e^{8\Delta} \sum_{\substack{A \subseteq [n] \\ |A| \geq 2}} \sum_{i \in A} |A| \cdot \Pr\brk{T \setminus \{i\} = A \setminus \{i\}}\cdot  \Rev_i\prn{D^T}\\
        &= 2e^{8 \Delta} \sum_{i \in [n]}  \Rev_i\prn{D^T} \sum_{\substack{A \ni i \\ |A| \geq 2}} |A| \cdot \Pr\brk{T \setminus \{i\} = A \setminus \{i\}} \\
        &= 2e^{8 \Delta} \sum_{i \in [n]}  \Rev_i\prn{D^T} \cdot \E\brk{|T \setminus \{i\}| \cdot \ind_{|T \setminus \{i\}| \geq 1}} \\
        &\leq 2e^{8 \Delta} \sum_{i \in [n]}  \Rev_i\prn{D^T} \cdot \E\brk{2 |T \setminus \{i\}|} \\
        &\leq 4e^{8 \Delta} \sum_{i \in [n]}  \Rev_i\prn{D^T} \cdot \E[|T|].
     \end{align*}
     
    To bound $\E[|T|]$, notice that $\Rev_i(D) \geq \Pr\brk{i \in T} \cdot e^{8\Delta} \SRev(D)$, since if $i \in T$, its contribution to $\SRev(D)$ is at least the lowest value in the tail, and thus
    \[
    \E[|T|] = \sum_i \Pr\brk{i \in T} \leq e^{-8\Delta} \sum_i \frac{\Rev_i(D)}{\SRev(D)} = e^{-8\Delta}.
    \]
    Therefore, we have 
    \[
    \sum_{\substack{A \subseteq [n] \\ |A| \geq 2}} \Pr\brk{T = A} \cdot \Rev\prn{D^T_A} \leq 4 \cdot \SRev\prn{D^T},
    \]
    so altogether we have $\sum_{A \subseteq [n]} \Pr\brk{T = A} \cdot \Rev\prn{D^T_A} \leq 5 \cdot \SRev\prn{D^T}$.
\end{proof}

Combining \Cref{lem:rev-tail} and \Cref{claim:coreAdditive} with \eqref{eq:addCoreTailDecomp} yields \Cref{thm:additive-srev-log-approx}.
\end{proof}

\subsection{\texorpdfstring{$\max\{\SRev, \BRev\}$}{max{SRev, BRev}} is \texorpdfstring{$O(\Delta)$}{O(Δ)}-approximate}

To get from $\bigO{\log{n} + \Delta}$ approximation  in \Cref{thm:additive-srev-log-approx} to $\bigO{\Delta}$ approximation in \Cref{thm:additive-rev}, the following \Cref{lem:addCoreRefined} refines our  core contribution bound in  \Cref{claim:coreAdditive}  in terms of $\BRev$ to save the $\log{n}$ factor. Together with tail contribution in \Cref{lem:rev-tail}, this gives the proof of \Cref{thm:additive-rev}.

\begin{lemma} \label{lem:addCoreRefined}
    The core contribution is
    \[
    \Val\prn{D^C} \leq \prn{22 \Delta + 1} \cdot \SRev(D) + 35 \prn{\Delta + 1} \cdot \BRev\prn{D}.
    \]
\end{lemma}
\begin{proof}
To get a refined bound on the core contribution, we will further split up the core. Let
\[
C_s \coloneqq \set{i \in C : g(\lt_i) \leq r}
\]
be the small elements of the core, and let
\[
C_\ell \coloneqq C \setminus C_s = \set{i \in C : r < g(\lt_i) < re^{8\Delta}}
\]
be the large elements. We define $D^{C_\ell}$ and $D^{C_s}$ as the restrictions of $D$ to $C_\ell$ and $C_s$, respectively. Since $C_s$ and $C_\ell$ partition $C$, we have 
\begin{align} \label{eq:coreCoreTail}
    \Val\prn{D^C} = \Val\prn{D^{C_s}} + \Val\prn{D^{C_\ell}}.
\end{align}
We bound each of these separately.

\begin{claim}
\[
\Val\prn{D^{C_\ell}} \leq 22 \Delta \cdot \SRev(D).
\]
\end{claim}
\begin{proof}
Consider the single item pricing strategy of picking a random $z \in \set{0, 1, \dots, 8\Delta - 1}$, and selling item $i$ for price $re^z$. If $i \in C_\ell$, then with probability $\frac{1}{8\Delta}$, item $i$ sells for at least a $\frac{1}{e}$ fraction of its value. Therefore, we have
\[
\Rev_i(D) \geq \frac{1}{8e\Delta} \cdot \E\brk{g(\lt_i) \cdot \ind_{i \in C_\ell}}.
\]
Summing over $i$ gives us that $\SRev(D) \geq \frac{1}{8e \Delta} \Val\prn{D^{C_\ell}} \geq \frac{1}{22 \Delta} \Val\prn{D^{C_\ell}}$.
\end{proof}

\begin{claim}
\[
\Val\prn{D^{C_s}} \leq \SRev(D) + 35 (\Delta+1) \cdot \BRev\prn{D^{C_s}}.
\]
\end{claim}
\begin{proof}
Let $\mu \coloneqq \Val\prn{D^{C_s}} = \E[v(C_s)].$ If $\mu \leq r$ we are done. Hence, assume that $\mu > r$.

Consider the following strategy for selling the grand bundle on $C_s$. Pick a random $z \in \set{-1, 0, 1, \dots, K}$ and offer the bundle at price $e^z \mu$, for a $K$ to be chosen later. Notice that if $v(C_s) \in \brk{e^{-1} \mu, e^K \mu}$, then this strategy obtains at least a $\frac{1}{e}$ fraction of the value of $v(C_s)$ with probability at least $\frac{1}{K + 2}$. Therefore, we find
\begin{align*}
\mu &\leq \mu \cdot e^{-1} + e (K+2) \BRev\prn{D^{C_s}} + \E\brk{\prn{v(C_s) - e^{K} \mu}^+}.
\end{align*}
We just need to bound the contribution of the last term. From Lemma 7 of \cite{cai-oikonomou-mrf}, we have
\begin{align} \label{eq:varBound}
\Var\prn{v(C_s)} ~\leq~ 2r^2 + \big(e^{4\Delta} - 1\big) \mu^2.
\end{align}
Therefore,
\begin{align*}
\E\brk{\prn{v(C_s) - e^K \mu}^+} ~=~ \int_{e^K \mu}^\infty \Pr\brk{v(C_s) \geq \tau} \dif \tau 
~\leq \int_{\prn{e^K - 1} \mu}^\infty \Pr\brk{v(C_s) - \mu \geq \tau} \dif \tau.
\end{align*}
Applying Chebyshev's inequality and bounding the variance by \eqref{eq:varBound} gives
\begin{align*}
\E\brk{\prn{v(C_s) - e^K \mu}^+}
~\leq~ \int_{\prn{e^K - 1} \mu}^\infty  \frac{2r^2 + (e^{4\Delta} - 1)\mu^2}{\tau^2} \dif \tau 
~=~ \frac{2r^2 + (e^{4\Delta} - 1)\mu^2}{\prn{e^K - 1}  \mu}.
\end{align*}
Next, we set $K = 4\Delta + 2$. Together with the fact that $\mu > r$, we have
\[
\E\brk{\prn{v(C_s) - e^K \mu}^+} ~\leq~ \frac{e^{4\Delta} + 1}{e^{4\Delta + 2} - 1} \cdot \mu \leq \frac{2 \mu}{e^2 - 1}.
\]
Altogether with our previous bounds, we finally obtain
\begin{align*}
\prn{1 - \frac{1}{e} - \frac{2}{e^2-1}} \mu &\leq e (4\Delta + 4) \cdot \BRev\prn{D^{C_s}} \\
\mu &\leq \frac{e (4\Delta + 4)}{\prn{1 - \frac{1}{e} - \frac{2}{e^2-1}}} \cdot \BRev\prn{D^{C_s}} \leq 35(\Delta+1) \cdot \BRev\prn{D^{C_s}}
\end{align*}
so $\Val\prn{D^{C_s}} \leq \SRev(D) + 35(\Delta+1) \cdot \BRev\prn{D^{C_s}}$ as desired.
\end{proof}
Combining the last two claims with \eqref{eq:coreCoreTail} completes the proof of the lemma.
 \end{proof}

\section{Beyond Additive Valuations}\label{sec:beyond-additive}
In this section we study unit-demand and subadditive buyers.

\subsection{Unit-Demand Buyer Mechanisms}\label{sec:unit-demand}

Recall, in the unit demand setting, the buyer samples a private MRF type vector $\bt$, and then the valuation for a subset  $S$  of items is
\[
v(S) = \max_{i \in S} g(t_i),
\]
where we employ the notation $g(t_i) = g(\{t_i\})$. 

\begin{proof} [Proof of \Cref{thm:ud-rev}]
Using the approximate marginal mechanism \Cref{lem:gen-splitting}, we may again use a core-tail decomposition. We define the tail by
\[
T \coloneqq \set{i \in [n] : g(t_i) \geq e^{8 \Delta+1} \SRev(D)},
\]
and set the core $C \coloneqq [n] \setminus T$. By applying \Cref{lem:gen-splitting}, we have
\begin{align} \label{eq:UnitCoreTailDecomp}
\Rev(D) = 2 \Val\prn{D^C} + 2\sum_{A \subseteq [n]} \Pr\brk{T = A} \cdot \Rev\prn{D^T_A}.
\end{align}

Unlike in the additive setting, the value of the core is already bounded by $\bigO{\Delta} \cdot \SRev(D)$. This allows us to simply follow the proof structure of \Cref{thm:additive-srev-log-approx} without losing an extra $\log n$ term.
\begin{claim}
For a unit demand buyer,
\[
\Val\prn{D^C} \leq \prn{22 \Delta + 4} \cdot \SRev(D).
\]
\end{claim}
\begin{proof}
Consider the following pricing strategy. Pick $z \in \{0, 1, \dots, 8\Delta\}$ uniformly at random, and sell every item for price $e^{z} \cdot \SRev(D)$. If $\max_{i \in C} g(t_i) \geq \SRev(D)$, then this strategy obtains revenue at least $\frac{v(C)}{e}$ with probability at least $\frac{1}{8\Delta + 1}$. Therefore, we have
\[
\E\brk{v_\bt(C)} \leq \SRev(v_\bt) + e(8\Delta +1) \SRev(v_\bt) \leq \prn{22 \Delta + 4} \SRev(v_\bt). \qedhere
\]
\end{proof}




\begin{claim}\label{claim:UnitTailBound}
For a unit demand buyer,
\[
\sum_{A \subseteq [n]} \Pr\brk{T = A} \cdot \Rev\prn{D^T_A} \leq 3 \cdot \SRev\prn{D^T}.
\]
\end{claim}
\begin{proof}[Proof of \Cref{claim:UnitTailBound}]
Notice that in the unit demand setting, we no longer have $\SRev(D) = \sum_i \Rev_i(D)$, as the buyer will not purchase more than one item. To get around this, we will show that an approximate version of this equality holds for $\SRev\prn{D^T}$, as it is rare that more than one item appears in the tail. Our goal will be to show 
\[
\SRev\prn{D^T} ~\gtrsim~ \sum_i\Rev_i\prn{D^T} ~\gtrsim~ \sum_{A \subseteq [n]} \Pr\brk{T = A} \Rev\prn{D^T_A},
\]
where  $\gtrsim$  means that the inequality holds up to scaling by a constant factor.

For the latter inequality, we again use the decomposition
\[
\sum_{A \subseteq [n]} \Pr\brk{T = A} \cdot \Rev\prn{D^T_A} = \sum_{i \in [n]} \Pr\brk{T = \{i\}} \cdot \Rev\prn{D^T_{\{i\}}} + \sum_{\substack{A \subseteq [n]\\|A| \geq 2}} \Pr\brk{T = A} \cdot \Rev\prn{D^T_A}.
\]
First, we clearly have $\Pr\brk{T = \{i\}} \cdot \Rev\prn{D^T_{\{i\}}} \leq \Rev_i\prn{D^T}$, so we only need to bound the contribution of the second summation. For this, we use \Cref{lem:gen-rev-n-srev} for unit-demand buyers to get
\begin{align*}
    \sum_{\substack{A \subseteq [n]\\|A| \geq 2}} \Pr\brk{T = A} \cdot \Rev\prn{D^T_A}
    &\leq \sum_{\substack{A \subseteq [n]\\|A| \geq 2}}  \Pr\brk{T = A} \cdot \sum_{i \in A} 4e^{4\Delta}  \Rev_i\prn{D^T_A} \\
    &= 4e^{4\Delta}  \sum_{i \in [n]} \sum_{\substack{A \ni i\\|A| \geq 2}}  \Pr\brk{T = A} \cdot \Rev_i\prn{D^T_A} \\
    &\leq 4e^{8\Delta}  \sum_{i \in [n]} \sum_{\substack{A \ni i\\|A| \geq 2}}  \Pr\brk{T \setminus \{i\} = A \setminus \{i\}} \cdot \Rev_i\prn{D^T} \\
    &= 4e^{8\Delta}  \sum_{i \in [n]} \Pr\brk{|T \setminus \{i\}| \geq 1} \cdot \Rev_i\prn{D^T} \\
    &\leq 4e^{8\Delta} \cdot \Pr\brk{|T| \geq 1} \cdot \sum_{i \in [n]}  \Rev_i\prn{D^T} .
\end{align*}
By construction, we have $\Pr\brk{|T| \geq 1} \leq e^{-8\Delta-1}$. Thus, we get
\[
\sum_{A \subseteq [n]} \Pr\brk{T = A} \cdot \Rev\prn{D^T_A} \leq \frac{4}{e} \sum_i \Rev_i\prn{D^T}.
\]
Now, we just need the bound $\SRev\prn{D^T} \gtrsim \sum_i \Rev_i\prn{D^T}$. Consider the mechanism of placing the optimal single-item price 
\[
p_i \coloneqq \argmax_{p \geq e^{8\Delta + 1} \SRev(D)} p \cdot \Pr\brk{g(\lt_i) \geq p}
\]
on each item $i$. This separate price mechanism generates revenue
\[
\sum_i p_i \cdot \Pr\brk{g(\lt_i) \geq p_i} \cdot \Pr\brk{i \text{ chosen} \midd g(\lt_i) \geq p_i} =  \sum_i \Rev_i(D^T) \cdot \Pr\brk{i \text{ chosen} \midd g(\lt_i) \geq p_i}.
\]
However, we have
\begin{align*}
\Pr\brk{i \text{ chosen} \midd g(\lt_i) \geq p_i} &\geq 1 - \Pr\brk{|T| \geq 2 \midd g(\lt_i) \geq p_i} \geq 1 - e^{4\Delta} \Pr\brk{|T \setminus \{i\}| \geq 1} \\
&\geq 1 - e^{-4\Delta-1}.
\end{align*}
Therefore, we have $\SRev\prn{D^T} \geq \prn{1 - e^{-4\Delta-1}} \sum_i \Rev_i\prn{D^T}$, and we obtain
\[
\sum_{A \subseteq [n]} \Pr\brk{T = A} \cdot \Rev\prn{D^T_A} \leq \frac{4}{e\prn{1 - e^{-8\Delta-1}}} \SRev\prn{D^T} < 3 \cdot \SRev\prn{D^T}.
\]
\end{proof}

Using the last two claims with \eqref{eq:UnitCoreTailDecomp} completes the proof of \Cref{thm:ud-rev}.
\end{proof}

\subsection{Subadditive Buyer Mechanisms}\label{sec:subadditive}

For the case of a subadditive buyer, since $\SRev(D)$ may be difficult to analyze directly, we will use the proxy $\SRev'(D)$, which we define as
\[
\SRev'(D) = \max_{(p_1, \dots, p_n) \in \R_+^n}  \sum_{i \in [n]} \E \Big[p_i \cdot \ind_{v(i) \geq p_i} \cdot \prod_{j \neq i} \ind_{v(j) < p_j} \Big] .
\]
In other words, $\SRev'(D)$ is the maximum expected revenue from a separate pricing mechanism where we are only allowed to collect revenue when the buyer purchases exactly one item. This is the same proxy used by \citet{rubinstein-weinberg-subadditive} in their analysis for a subadditive buyer with independent items. Clearly, $\SRev'(D) \leq \SRev(D)$ and $\SRev'(D) \leq \sum_i \Rev_i(D)$. Hence, \Cref{thm:subadd-rev} is  an immediate corollary of the following lemma.

\begin{lemma} \label{thm:subadditive}
For a single subadditive buyer, we have
\[
\Rev(\bt) \leq \prn{348 \Delta + 110} \cdot \BRev(\bt) + 10 \cdot \SRev'(\bt).
\]
\end{lemma}

We will let $t > 0$ denote the cut-off between the tail and core, which we will define momentarily:
$$T := \{i \in [n] : v(i) \geq t\}, \hspace{0.1\textwidth} C:= [n] \setminus T.$$
For each $i$, let $q_i := \Pr(i \in T)$. We choose $t$ such that $\sum_{i \in [n]} q_i = e^{-8\Delta-1}$. Notice that this gives the bound $t \leq e^{8 \Delta + 2} \SRev'(D)$ by 
\begin{align*}
    \SRev'(D) &\geq \sum_i t \cdot \Pr(v(i) \geq t \And v(j) < t,\; \forall j \neq i) \\
    &= \sum_i t \cdot \Pr(v(i) \geq t) \cdot \Pr\brk{v(j) < t, \; \forall j \neq i \mid v(i) \geq t}\\
    &\geq  \sum_i t \cdot \Pr(v(i) \geq t) \cdot \Big(1 - \Pr(\exists j \neq i, \; v(j) \geq t \mid v(i) \geq t) \Big) \\
    &\geq t \cdot \left(\sum_i q_i \right) \left(1-e^{4\Delta}\sum_j q_j\right) \\
    &\geq t e^{-8\Delta-2}.
\end{align*}

\begin{lemma}
    We have
    \[
    \sum_{A \subseteq [n]} \Pr(T = A) \cdot \Rev(D^T_A) \leq 4 \cdot \SRev'(D).
    \]
\end{lemma}
\begin{proof}
Again, we separately consider when $|T| = 1$ and when $|T| \geq 2$.
\begin{align*}
\sum_{A \subseteq [n]} \Pr\brk{T = A} \Rev\prn{D^T_A} &=
\sum_{|A| \geq 2} \Pr\brk{T = A} \Rev\prn{D^T_A} + \sum_{i \in [n]} \Pr\brk{T = \{i\}} \Rev_{\{i\}}\prn{D^T_{\{i\}}} \\
&\leq \sum_{|A| \geq 2} \Pr\brk{T = A} \Rev\prn{D^T_A} + \SRev'\prn{D^T}.
\end{align*}
For the terms with $|A| \geq 2$, we can use an argument similar to that of \Cref{lem:rev-tail} for the additive setting. Specifically, we let $p$ be the optimal price in the mechanism for $\Rev_i(D^T_A)$, and we have
\begin{align*}
    \Rev_i(D^T) &\geq p \cdot \Pr\brk{v(i) \geq p} \\
    &\geq e^{-4\Delta} \cdot p \cdot \Pr(v(i) \geq p \mid T \setminus \{i\} = A \setminus \{i\})\\
    &= e^{-4\Delta} \cdot p \cdot \Pr(v(i) \geq p \mid T = A) \cdot \Pr(T = A \mid T \setminus \{i\} = A \setminus \{i\}) \\
    &= e^{-4\Delta} \Rev_i(D^T_A) \cdot \frac{\Pr(T = A)}{\Pr{T \setminus \{i\} = A \setminus \{i\}}}.
\end{align*}

Thus, we see that $\Pr[T = A] \cdot \Rev_i(D^T_A) \leq e^{4\Delta} \Pr\brk{T \setminus \{i\} = A \setminus \{i\}} \cdot \Rev_i(D^T)$. Using this and \Cref{lem:gen-rev-n-srev} we have
\begin{align*}
\sum_{\substack{A \subseteq [n]\\|A| \geq 2}} \Pr\brk{T = A} \cdot \Rev\prn{D^T_A}
&\leq 2e^{4\Delta} \cdot 
\sum_{\substack{A \subseteq [n]\\|A| \geq 2}} \Pr\brk{T = A} \cdot |A| \cdot \sum_{i \in A} \Rev_i(D^T_A)\\
&\leq 2e^{8\Delta} \cdot 
\sum_{\substack{A \subseteq [n]\\|A| \geq 2}} \sum_{i \in A} |A| \cdot \Pr\brk{T \setminus \{i\} = A \setminus \{i\}} \cdot \Rev_i(D^T)\\
&= 2e^{8\Delta} \cdot 
\sum_{i \in [n]} \Rev_i(D^T) \cdot \sum_{\substack{A \ni i \\|A| \geq 2}} |A| \cdot \Pr\brk{T \setminus \{i\} = A \setminus \{i\}} .
\end{align*}
Now, notice that for any $i \in [n]$, we have
\begin{align*}
    \sum_{\substack{A \ni i\\|A| \geq 2}} |A| \cdot \Pr\brk{T \setminus \{i\} = A \setminus \{i\}} = \E\brk{(|T \setminus \{i\}| + 1) \ind_{|T \setminus \{i\}| \geq 1}} \leq \E\brk{2 \cdot |T \setminus \{i\}|} \leq 2 \E[|T|].
\end{align*}
Thus, we ultimately have
\begin{align*}
\sum_{\substack{A \subseteq [n]\\|A| \geq 2}} \Pr\brk{T = A} \cdot \Rev\prn{D^T_A}
~\leq~ 4e^{8\Delta} \cdot \E[|T|] \cdot
\sum_{i \in [n]} \Rev_i(D^T) < 2
\sum_{i \in [n]} \Rev_i(D^T),
\end{align*}
where the last inequality comes from the fact that $\E[|T|] \leq e^{-8\Delta - 1}$. Finally, we can argue that $\sum_i \Rev_i\prn{D^T} \lesssim \SRev'\prn{D^T}$ similarly to our argument in the unit demand setting. Letting $p_i$ be the optimal price of $i$ in the mechanism for $\Rev_i\prn{D^T}$, we have
\begin{align*}
\SRev'\prn{D^T} &\geq \sum_i p_i \cdot \Pr\brk{v(i) \geq p_i} \cdot \Pr\brk{\forall j \neq i,\;v(j) < p_j \midd v(i) \geq p_i} \\
&\geq \sum_i p_i \cdot \Pr\brk{v(i) \geq p_i} \cdot \prn{1 - \sum_{j \neq i}\Pr\brk{v(j) \geq p_j \midd v(i) \geq p_i}} \\
&\geq \sum_i p_i \cdot \Pr\brk{v(i) \geq p_i} \cdot \prn{1 - e^{4\Delta} \E[|T|]} \geq \sum_i \Rev_i\prn{D^T} \cdot \prn{1 - e^{-4\Delta-1}} \\
&\geq \frac{1}{2} \cdot \sum_i \Rev_i\prn{D^T}. \qedhere
\end{align*}
\end{proof}

\begin{lemma}\label{lem:subadd-core-brev}
We have
\[
\Val(D^C) \leq \prn{174 \Delta + 55} \cdot \BRev(D^C) + \SRev'(D).
\]
\end{lemma}
\begin{proof}
We may assume $\Val(D^C) \geq \SRev'(D)$, or else we are done. Let $\mu = \Val(D^C)$. We use the following pricing strategy for the grand bundle: select $z \in \{-1, 0, 1, \dots, K\}$ uniformly at random, where $K = O(\Delta)$ is chosen later, and sell the grand bundle at price $\mu e^z$. By standard analysis, we see that this obtains revenue at least $\frac{1}{eK} \E \brk{\ind_{v(C) \geq \mu e^{-1}} \cdot \min\{v(C), \mu e^K\}}$. Additionally, we have the bound
\begin{align}
   \mu &\leq \mu e^{-1} + \E \brk{\ind_{v(C) \geq \mu e^{-1}} \cdot \min\{v(C), \mu e^K\}} + \E(v(C) - \mu e^K)^+,\\
   \label{eq:subadd-core} \implies (1 - e^{-1}) \mu &\leq \E \brk{\ind_{v(C) \geq \mu e^{-1}} \cdot \min\{v(C), \mu e^K\}} + \E(v(C) - \mu e^K)^+.
\end{align}

Hence, we only need to show that $\E(v(C) - \mu e^K)^+ \leq c \cdot \mu$ for some small constant $c$ to get our desired result. To do this, we seek to use a concentration bound for subadditive functions over independent items.

First, we will reduce to the independent case. Let $\bt^\indep \in \prod_{i \in [n]} (\Omega_i \cup \{0\})$ be a random vector with independent coordinates with marginal distributions on each $\lt^\indep_i$ given as follows.
\begin{align*}
    \Pr(\lt^\indep_i = \omega_i) 
    &= \inf_{\omega_{-i} \in \Omega_{-i}}\Pr(t_i = \omega_i \mid \bt_{-i} = \omega_{-i}) && \forall i \in [n],~\omega_i \in \Omega_i,\\
    \Pr(\lt^\indep_i = 0) &= 1 - \sum_{\omega \in \Omega_i} \Pr(t^\indep_i = \omega).
\end{align*}

Here, $0$ is a dummy element that we introduce, for which we define $g(0) = 0$, i.e. so that $g(S \cup \{0\}) = g(S)$ for any $S \subseteq \Omega$.

\begin{claim}\label{clm:subadd-to-indep}
For a subadditive monotone function $g : \Omega\to \R_+$, let $g(\bt) := g(\{t_i : i \in [n]\})$. Then, for each $\tau > 0$,
\[
\Pr(g(\bt) \geq \tau) \leq e^{4\Delta} \cdot \Pr(g(\bt^\indep) \geq e^{-4\Delta} \tau).
\]
\end{claim}
\begin{proof}
    Let $S$ be a random subset of $[n]$ in which each $i$ is included independently with probability $e^{-4\Delta}$. We use the notation $g(\bt_{S}) = g(\{t_i : i \in S\})$.
    
    We claim that $g(\bt^\indep)$ stochastically dominates $g(\bt_{S})$, in the sense that we can define a coupling of $\bt^\indep$ with $\bt$ and $S$ such that $i \in S$ only when $\lt^\indep_i = \lt_i$. This ensures that $g(\bt^\indep) \geq g(\bt_{S})$ under the coupling, and hence the cumulative distribution function of $g(\bt^\indep)$ dominates that of $g(\bt)$.

    The coupling we use is as follows. For $i = 1, \dots, n$, we iteratively sample $\lt_i$ from $\Omega_i$, conditional on our previous observation $\lt_1, \dots, \lt_{i-1} = \omega_1, \dots, \omega_{i-1}$, and then independently decide if $i \in S$. In the case $i \in S$, we choose $\lt_i^\indep = \lt_i$. In the remaining case $i \not \in S$, we sample $\lt_i^\indep$ with appropriate probabilities to have the correct marginal probabilities $\Pr(\lt_i^\indep = \omega_i)$. Note that this is possible since Lemma 1 of \cite{cai-oikonomou-mrf} gives us
    \[
    \Pr(\lt^\indep_i = \omega_i) \geq e^{-4\Delta} \Pr(\lt_i = \omega_i \mid \bt_{<i} = \omega_{<i}) = \Pr(i \in S \text{ and } \lt_i = \omega_i \mid \bt_{<i} = \omega_{<i}).
    \]

Using this property of $\bt^\indep$, lets consider a uniform random map\footnote{We assume that $e^{4\Delta}$ is in integer for simplicity, as we can round up $\Delta$ to the nearest value for which this holds to obtain the same asymptotic bounds.} $\sigma : [n] \to [e^{4\Delta}]$. Notice that for each $k$, the set $\sigma^{-1}(k)$ has the distribution of the random set $S$. Thus, for any $\tau > 0$ we have
\begin{align*}
    \Pr(g(\bt) \geq \tau) &\leq \Pr(\exists k,\; g(\bt_{\sigma^{-1}(k)}) \geq e^{-4\Delta} \tau) \\
    &\leq e^{4\Delta} \Pr(g(\bt_{S}) \geq e^{-4\Delta} \tau)\\
    &\leq e^{4\Delta} \Pr(g(\bt^\indep) \geq e^{-4\Delta} \tau). \qedhere
\end{align*}
\end{proof}

Using this reduction, we can use the subadditive concentration theorem for product distributions to get an approximate concentration result for MRFs. From Theorem 3.10 of \cite{rubinstein-weinberg-subadditive}, we have that if $a$ is the median of $g(\bt^\indep)$ and $L$ is the Lipschitz constant of $g$, then
$
\Pr(g(\bt^\indep) \geq 3a + \tau) \leq 4 \cdot 2^{-\tau / L}.
$
Letting $\mu = \E[g(\bt)]$, we have $a \leq 2 \E[g(\bt^\indep)] \leq 2 \mu$. Combined with the concentration result above, this tells us
\[
\Pr[g(\bt^\indep) \geq 6\mu + \tau] \leq 4 \cdot 2^{-\tau / L}.
\]

Hence, for any $\tau_0 > 0$, we compute
\begin{align*}
    \E[\prn{g(\bt) - \tau_0}^+] &= \int_{\tau_0}^\infty \Pr(g(\bt) \geq \tau) d\tau\\
    &\leq e^{4\Delta} \int_{\tau_0}^\infty  \cdot \Pr(g(\bt^\indep) \geq e^{-4\Delta} \tau) d\tau\\
    &= e^{8\Delta} \int_{e^{-4\Delta}\tau_0}^\infty \Pr(g(\bt^\indep) \geq \tau) d\tau\\
    &\leq 4e^{8\Delta} \int_{e^{-4\Delta}\tau_0 - 6\mu}^\infty 2^{-\tau / L} d\tau\\
    &= \frac{4L}{\ln 2}\exp\prn{8\Delta - \ln(2) \cdot \frac{e^{-4\Delta}\tau_0 - 6\mu}{L}}.
\end{align*}

Now, to apply this to our setting with $v(C) = g(\bt_C)$, notice that $g$ is $t$-Lipschitz on the core elements. Hence, setting $\tau_0 = e^{16\Delta + 5} \mu$ and recalling $t \leq e^{8\Delta + 2} \cdot \SRev'(D) \leq e^{8\Delta + 2} \mu$, we have
\begin{align*}
\E\brk{g(\bt_C) - e^{16 \Delta + 5} \mu)^+} &\leq \frac{4e^{8 \Delta + 2} \mu}{\ln 2} \cdot \exp\prn{8\Delta - \ln(2) \cdot \frac{e^{12\Delta +5}  - 6}{e^{8 \Delta + 2}}} \\
&\leq \mu \cdot \frac{4e^2}{\ln 2} \cdot \exp\prn{16\Delta - \ln(2) \cdot \frac{e^{12 \Delta +3}  - \frac{6}{e^2}}{e^{8 \Delta}}} \\
&\leq \mu \cdot \frac{4e^2}{\ln 2} \cdot \exp\prn{16\Delta - 7 \ln(2) \cdot e^{4\Delta}} \\
&\leq \mu \cdot \frac{4e^2}{\ln 2} \cdot \frac{1}{128} \\
&\leq \frac{\mu}{3}.
\end{align*}
Thus, taking $K = 16\Delta + 5$ and substituting the above into \eqref{eq:subadd-core}, we obtain
\[
\BRev(D^C) \geq \frac{1}{e\prn{16 \Delta + 5}} \E\brk{\ind_{v(C) \geq \mu e^{-1}}} \geq \frac{\mu}{4 e\prn{16 \Delta + 5}} \geq \frac{1}{174 \Delta + 55} \cdot \mu.
\]
\end{proof}
\section{Prophet Inequalities}\label{sec:prophet}

In this section, we study the prophet inequality problem under MRF dependencies. Formally, in the prophet inequality problem, there are $n$ non-negative random variables $(X_1, \ldots,X_n)$ whose values are drawn from a known joint distribution $D$.  
In the $i$-th step for $i\in [n]$, the online algorithm sees $X_i$ (but $X_j$ for $j>i$ are still unknown) and has to immediately decide to accept/reject $X_i$. The game ends when the algorithm first accepts an element $X_\tau$ and the algorithm's goal is to maximize $\E[X_\tau]$. We say an algorithm is $\alpha$-competitive if $\E[\max_i X_i] \leq \alpha \cdot \E[X_\tau]$.

\begin{definition}
We say that a distribution $D$ over $\R_+^n$ has an MRF distribution 
if there exists an MRF  $\mrf$  (recall \Cref{def:mrf}) and functions $g_i : \Omega_i \to \R_+$ for each $i \in [n]$ such that $(X_1, \dots, X_n) \sim D$, where $X_i = g_i(t_i)$ and $\bt$ is sampled from $\mrf$.
\end{definition}

In \Cref{sec:delta-prophet-lower}, we will  show an $O(\Delta)$-competitive algorithm for prophet inequality, exponentially improving the results of \citet{cai-oikonomou-mrf}. 
In \Cref{sec:delta-prophet-upper}, we show that this bound is tight (up to constant factors) for any online algorithm under MRF distributions. 
Finally, in \Cref{apx:ocrs} we  show that the popular approach of designing prophet inequalities for product distributions  via ``online contention resolution schemes'' (OCRSs) cannot be used to prove \Cref{thm:prophet}. In particular, it is impossible to obtain better than $\exp(\Theta(\Delta))$-selectable OCRS under MRF dependencies, hence showing a strong separation between prophet inequalities and OCRS under MRF dependencies.
 
\subsection{\texorpdfstring{$\bigO{\Delta}$}{O(Δ)}-competitive Prophet Inequality}\label{sec:delta-prophet-lower}

Our algorithm in the proof of \Cref{thm:prophet} is simple: it selects a single threshold $\tau$ chosen at random from a geometric distribution, and accepts the first random variable above the threshold. By carefully balancing between the chance of guessing the threshold for the maximum and the probability that more than one random variables exceed the upper-most limit of the geometric distribution, we obtain an $O(\Delta)$-competitive algorithm. 

\begin{proof}[Proof of \Cref{thm:prophet}]
    Let $M := \max_i X_i$. Without loss of generality, assume $\Opt := \E[M] = 1$.     
    Our strategy will be as follows: First, pick a value $z \in \{-1, 0, 1, 2, \dots, K\}$ uniformly at random, where $K$ will be chosen later, and then take the first value $X_i$  which is at least $e^{z}$. Let $\alg$ be the random value taken by this algorithm.
    Now, we upper bound the maximum as
    \[ \textstyle
        1 = \E[M] \leq e^{-1} + \sum_{i = 0}^{K} \Pr\brk{M \in [e^{i-1}, e^i)}\cdot e^i + \sum_{i \in [n]} \E\brk{\ind_{X_i \geq e^{K}} \cdot X_i},
    \]
    which implies
    \[     \textstyle    (1-e^{-1}) \E[M] \leq \sum_{t = 0}^{K} \Pr\brk{M \in [e^{i-1}, e^i)}\cdot e^i + \sum_{i \in [n]} \E\brk{\ind_{X_i \geq e^{K}} \cdot X_i}.
    \]

    At the same time, we can lower bound our algorithm as
    \begin{align*}
        \E[\alg] &\geq \sum_{t = 0}^{K} \Pr\brk{M \in [e^{i-1}, e^i) \wedge z = i-1} e^{i-1} + \sum_{i \in [n]} \Pr\brk{z = K} \E\brk{\ind_{X_i \geq e^{K}} \ind_{\bX_{-i} < e^{K}} X_i} \\
        &=\frac{1}{K + 2} \prn{ e^{-1}\sum_{t = 0}^{K} \Pr\brk{M \in [e^{i-1}, e^i)} \cdot e^{i} + \sum_{i \in [n]} \E\brk{\ind_{X_i \geq e^{K}} \cdot \ind_{\bX_{-i} < e^{K}} \cdot X_i} }.
    \end{align*}

    Next, we show that $\E\brk{\ind_{X_i \geq e^K} \cdot \ind_{\bX_{-i} < e^{K}} \cdot X_i}$ is at least a constant fraction of $\E\brk{\ind_{X_i \geq e^K} \cdot X_i}$. Notice that
    \begin{align*}
        \E\brk{\ind_{X_i \geq e^K} \cdot \ind_{\bX_{-i} < e^{K}} \cdot X_i} &= \int_{\substack{x_i \geq e^K \\ \bx_{-i} < e^{K}}} x_i \dif \Pr\brk{\bX = \bx} \\
        &= \int_{x_i \geq e^K} x_i \Pr\brk{\bX_{-i} < e^{K} \midd X_i = x_i} \cdot \dif \Pr\brk{\bX = \bx} \\
        &= \int_{x_i \geq e^K} x_i \prn{1 - \Pr\brk{\exists j \neq i,\; X_j > e^{K} \midd X_i = x_i}} \dif \Pr\brk{X_i = x_i} \\
        &\geq \int_{x_i \geq e^K} x_i \prn{1 - e^{4\Delta} \Pr\brk{\exists j \neq i,\; X_j > e^{K}}} \cdot \dif \Pr\brk{X_i = x_i} \\
        &\geq \int_{x_i \geq e^K} x_i \prn{1 - e^{4\Delta-K}} \cdot \dif \Pr\brk{X_i = x_i} \\
        &= \prn{1 - e^{4 \Delta - K}} \E\brk{\ind_{X_i \geq e^K} \cdot X_i},
    \end{align*}
    where the last inequality follows from $\Pr\brk{\exists j \neq i,\; X_j > e^{K}} \leq \frac{\E[M]}{e^K} = e^{-K}$ which is due to Markov's inequality. Choosing $K = 4\Delta + 1$ gives
    \begin{align*}
    \E[\alg] &\geq \frac{1}{K + 2} \prn{ e^{-1}\sum_{t = 0}^{K} \Pr\brk{M \in [e^{i-1}, e^i)} \cdot e^{i} + \prn{1 - e^{-1}} \sum_{i \in [n]} \E\brk{\ind_{X_i \geq e^{K}} \cdot X_i} } \\
    &\geq \frac{e^{-1}(1-e^{-1})}{4\Delta + 3} \geq \frac{1}{20 \Delta + 15}.
    \end{align*}
\end{proof}

\subsection{Every algorithm is \texorpdfstring{$\bigOm{\Delta}$}{Ω(Δ)}-competitive}\label{sec:delta-prophet-upper}

Next, we show that one cannot hope for a prophet inequality with a better than linear dependence on $\Delta$. 

\begin{theorem}\label{thm:pi-lin-upper}
For any $\Delta > 0$, there exists an MRF with maximum weighted degree $\Delta$ such that for the prophet inequality setting with this MRF, no online algorithm is better than $\frac{\Delta+1}{2}$-competitive.
\end{theorem}

Our hardness example consists of an MRF that is a path, i.e., the dependency graph looks like that of \Cref{fig:path-mrf}. In such MRFs, the value of $X_i$ depends only on the value of $X_{i-1}$, for any $i$, and thus, indirectly on all $X_j$ for $j < i$. For such MRFs, which inherently impose an ordering on the vertices, we present a lemma regarding Markov Random Fields that will be useful in our proof. 
\begin{figure}
    \centering
    \includegraphics[width=0.5\textwidth]{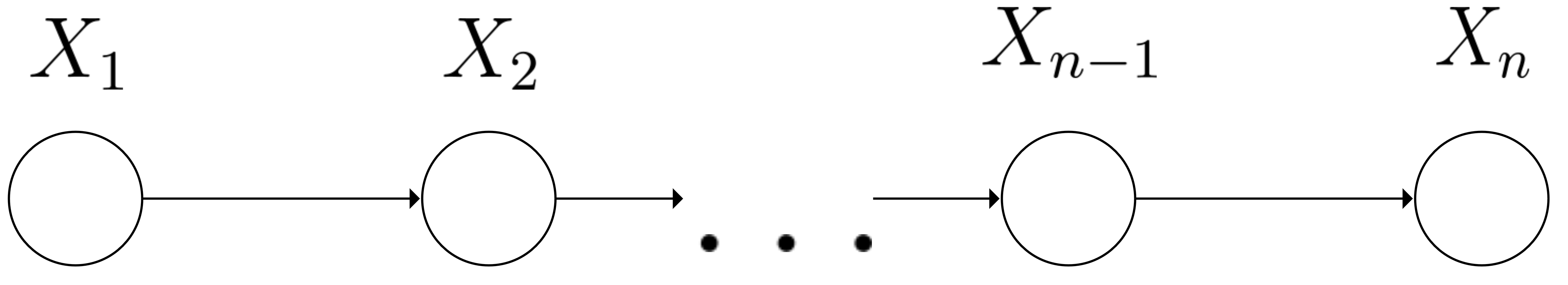}
    \caption{The dependency graph of a path MRF.}
    \label{fig:path-mrf}
\end{figure}

\begin{lemma}\label{lem:mrf-choose-probs}
For any $i \in [n]$, any probabilities $\{p_{i,\omega} \in (0,1): \omega \in \Omega_i\}$ such that $\sum_{\omega \in \Omega_i} p_{i, \omega} = 1$ and any fixed functions $\{\psi_e: e \in E\}$, there exists a path MRF $\mrf = \prn{\{\Omega_i\}_{i \in [n]},\; E,\; \{\psi_i\}_{i \in [n]},\; \{\psi_e\}_{e \in E}}$ where $E = \set{\set{i-1, i} \midd i \in [n]}$, such that if $(X_1, \dots, X_n)$ follows $\mrf$, then
\[
\forall i \in [n],\; \omega \in \Omega_i,\; \bx = \prn{x_1, \dots, x_n} \in \Omega_1 \times \dots \times \Omega_n, \qquad \Pr\brk{X_i = \omega \midd \forall j < i,\; X_j = x_j} = p_{i, \omega}.
\]
\end{lemma}

\begin{proof}[Proof of \Cref{lem:mrf-choose-probs}]
Notice that $\Pr\brk{X_i = x_i \midd \forall j < i,\; X_j = x_j}$ does not depend on the functions $\psi_1, \dots, \psi_{i-1}$. This allows us to choose the $\psi_i$ sequentially in reverse order starting from $i = n$ and proceeding to $i = 1$.

We proceed by induction. Suppose for some $i$, we have already fixed $\psi_{i+1}, \dots, \psi_n$. For each $\omega \in \Omega_i$, let $\psi_i(\omega) = z_\omega$ for some $z_\omega$ we will choose later. We also adopt the notation $e_{< i} := e \cap \{1, \dots, i-1\}$ and $e_{> i} := e \cap \{i+1, \dots, n\}$ for $e \in E$. With this in mind, we have
\begin{align*}
&\Pr\brk{X_i = \omega \midd \forall j < i,\; X_j = x_j} \\
&\qquad = \frac{\sum_{u_{i+1}, \dots, u_n} \exp \prn{ \psi_i(\omega) + \sum_{j=i+1}^n \psi_j(u_j) + \sum_{e \in E} \psi_e(x_{e_{< i}}, \omega, u_{e_{>i}}) } }{\sum_{\omega', u_{i+1}, \dots, u_n} \exp \prn{ \psi_i(\omega') + \sum_{j=i+1}^n \psi_j(u_j) + \sum_{e \in E} \psi_e(x_{e_{< i}}, \omega', u_{e_{>i}}) } }, \\
&\qquad = \frac{\lambda_\omega e^{z_\omega}}{\sum_{\omega'}\lambda_{\omega'} e^{z_{\omega'}}},
\end{align*}
where $\lambda_{\omega'} = \sum_{u_{i+1}, \dots, u_n} \exp \prn{ \sum_{j=i+1}^n \psi_j(u_j) + \sum_{e \in E} \psi_e(x_{e_{< i}}, \omega', u_{e_{>i}}) }$ for each $\omega' \in \Omega_i$. Thus, choosing $z_\omega = \log\prn{\f{p_{i,\omega}}{\lambda_\omega}}$ gives the desired result.  
\end{proof}

Given the lemma, we can prove \Cref{thm:pi-lin-upper}.

\begin{proof}[Proof of \Cref{thm:pi-lin-upper}]
Let $\lambda_1, \lambda_2, \dots$ and $0 < p < 1$ be quantities to be chosen later. The MRF we will construct is a path MRF, like the one shown in Figure~\ref{fig:path-mrf}. In other words, $E = \set{\{i-1, i\} \colon i \in [n]}$. For each positive $n$, we define a MRF prophet inequality problem for $\Omega = \Omega_0 \times \dots \times \Omega_n$ by setting
\begin{align*}
    \Omega_0 &:= \{1\}, \\
   \Omega_i &:= \{0, 1\} &&\text{ for } 1 \leq i \leq n,\\
    X_i = g_i(\lt_i) &:= \lt_i \cdot \prod_{j = 0}^{i-1}\lambda_{n-j} &&\text{ for } 0 \leq i \leq n, \\
    \psi_{i-1, i}(\lss_{i-1}, \lss_i) &= \begin{cases}
    -\Delta & \text{if } \lss_{i-1} \neq \lss_i\\
    \Delta & \text{if } \lss_{i-1} = \lss_i.
    \end{cases}. &&\text{ for } 1 \leq i \leq n.
\end{align*}
Additionally, by Lemma~\ref{lem:mrf-choose-probs}, we can choose $\psi_i(x)$ for $i \in [n]$ (in descending order of $i$) so that $\Pr(t_i = 0 \mid t_{i-1} = 1) = p$ (notice that the path structure ensures that $t_i$ is independent of $t_0, \dots, t_{i-2}$ given $t_{i-1}$).

Let $R^{(1)}_n$ be the optimal reward obtainable on the above prophet inequality setup for a given $n$. Moreover, let $R^{(0)}_n$ be the optimal reward obtainable if we instead define $\Omega_0 = \{0\}$. Notice that we have the following recursive properties:
\begin{align*}
    R^{(1)}_n = \max \set{1,\; \lambda_n \cdot \prn{ (1-p) R^{(1)}_{n-1} + p R^{(0)}_{n-1} } } \: \: \text{and} \: \:
    R^{(0)}_n = \lambda_n \cdot \prn{ qR^{(1)}_{n-1} + (1-q)R^{(0)}_{n-1} },
\end{align*}
where $\frac{q}{1-q} = e^{-4\Delta} \cdot \frac{1-p}{p}$

We seek to choose $p$ and $\lambda_n$ so that the $R^{(1)}_n = \lambda_n \cdot \big( (1-p) R^{(1)}_{n-1} + p R^{(0)}_{n-1}\big)$ for all $n$. If this is the case, then defining the vector $R_n = \brk{R^{(1)}_n, R^{(0)}_n}^T \in \R^2$ gives
$
R_n = \lambda_n \cdot \begin{bmatrix}
(1-p) & p\\
q & (1-q)
\end{bmatrix} R_{n-1}.
$
Hence,
\[
R_n ~=~ \prod_{i=1}^n \lambda_i \cdot \begin{bmatrix}
(1-p) & p \\
q & (1-q)
\end{bmatrix}^n
\begin{bmatrix}
1 \\ 0
\end{bmatrix}
~~=~~
\prod_{i=1}^n \lambda_i \cdot \prn{
\frac{q}{p + q}
\begin{bmatrix}
1 \\ 1
\end{bmatrix}
+
\frac{(1-p-q)^n}{p + q}
\begin{bmatrix}
p \\ -q
\end{bmatrix}
}.
\]
This leads us to choose $\lambda_n = \frac{q + p(1-p-q)^{n-1}}{q + p(1-p-q)^{n}}$ so that $R_n^{(1)} = 1$ for all $n$.

Now, let us consider quantities $M^{(1)}_n$ and $M^{(0)}_n$, which we define to be expected value of $\max_{i \in [n]} g_i(X_i)$ for the prophet inequality problem above on a given $n$ in the cases when $\Omega_0 = \{1\}$ and when $\Omega_0 = \{0\}$ respectively. Our goal is to show that $M^{(1)}_n$ is large for some $n$.

Similar to the recursion for $R_n$, we have the following recursive formula for $M_n = \prn{M^{(1)}_n, M^{(0)}_n}$.
\begin{align}
M_n &= \lambda_n \cdot \begin{bmatrix}
(1-p) & p \\
q & (1-q)
\end{bmatrix} M_{n-1} + p(1-q)^{n-1} 
\begin{bmatrix}
1 \\ 0
\end{bmatrix}, \nonumber \\
M_n &= R_n + p \sum_{k=1}^n (1-q)^{k-1} \cdot \frac{\prod_{i = {k+1}}^n \lambda_i}{\prod_{i = 1}^{n-k} \lambda_i} \cdot R_{n-k} \quad , \label{eq:max-formula} 
\intertext{which implies}
 M^{(1)}_n  &= 1 + p \sum_{k = 1}^n (1-q)^{k-1} \cdot \frac{\prod_{i = {k+1}}^n \lambda_i}{\prod_{i = 1}^{n-k} \lambda_i} \nonumber \\
&= 1 + p \sum_{k = 1}^n (1-q)^{k-1} \cdot \frac{\prn{q + p(1-p-q)^k}\prn{q + p(1-p-q)^{n-k}}}{\prn{q + p(1-p-q)^n}\prn{q + p}}. \label{eq:pi-main}
\end{align}

Notice that for $\Delta \to +\infty$, we have $q = 0$ and $\lambda_i = \f{1}{p}$ for all $i$. Thus, $M^{(1)}_n \approx n$ for small $p$, since each term in the sum from \eqref{eq:max-formula} is $1$.

For finite $\Delta$, we seek to approximate the phenomenon that all terms in the sum contribute about $1$. However, we can only do this if $n$ is not too large. Notice that taking $n \to \infty$ gives
\begin{align*}
M^{(1)}_\infty ~=~ 1 + p \sum_{k = 1}^\infty (1-q)^{k-1} \cdot \frac{1}{\prod_{i = 1}^{k} \lambda_i} ~&=~ 1 + p \sum_{k = 0}^\infty (1-q)^k \frac{q + p (1-p-q)^{k+1}}{q+p}\\
&= 1 + \frac{p}{q+p} \prn{1 + \frac{p(1-p-q)}{1 - (1-+p-q)(1-q)}} ~\leq~ 3.
\end{align*}

Instead, we set $p = \frac{1}{2}$ so that $q = \frac{e^{-4\Delta}}{1 + e^{-4\Delta}}$. Setting $n = \ceil{\log_{\f{1}{2} - q}(2q)} \leq 6 \Delta$ gives
\begin{align*}
M^{(1)}_n &= 1 + \f{1}{2} \cdot \sum_{k = 1}^n (1-q)^{k-1} \cdot \frac{\prn{2q + \prn{\f{1}{2} - q}^k}\prn{2q + \prn{\f{1}{2} - q}^{n-k}}}{\prn{2q + \prn{\f{1}{2}  -q}^n} \prn{2q + 1}}\\
&\geq 1 + \f{1}{2} \cdot \sum_{k = 1}^n (1-q)^{k-1} \cdot 
\frac{\prn{\f{1}{2} - q}^n}{4 q (2q + 1)} \\
&\geq 1 + \f{1}{2} \cdot \sum_{k = 1}^n (1-q)^{k-1} \cdot \frac{2q \prn{\f{1}{2} - q}}{4q (2q + 1)}\\
&\geq 1 + \frac{1-2q}{8 (1 + 2q)} \sum_{k = 1}^n (1-q)^{k-1} \\
&\geq 1 + (n-1) \cdot \frac{(1 - 2q)(1 - qn)}{8 (1 + 2q)} \\
&\geq \frac{\log\prn{\frac{64}{q}}}{8} \geq \frac{\Delta + 1}{2}. \qedhere
\end{align*}
\end{proof}

\subsection{Online Contention Resolution Schemes}\label{apx:ocrs}

Here we show that one cannot hope to use OCRS to obtain a prophet inequality with a polynomial dependence on $\Delta$. First, we formally introduce Online Contention Resolution Schemes.

\begin{definition}[Online Contention Resolution Scheme (OCRS) \cite{ocrs}]
Let $\cI$ be a feasibility constraint over $[n]$. For an online selection setting where we are given a distribution $\cD$ over $[n]$ and a point $\bx \in \cP_\cI$ for a polyhedral relaxation $\cP_\cI$ of $\cI$, we draw a random subset of the elements $R(\bx)$ according to $\cD$, where the marginal probability that $i$ appears in $R$ is $x_i$. We call $R(\bx)$ the set of \emph{active} elements. Afterwards, we observe whether each element $i \in [n]$ is active ($i \in R(\bx)$), one by one, and have to immediately and irrevocably decide whether to select an element or not before the next element is revealed. An {\em Online Contention Resolution Scheme} $\pi$ for $\cP$ is an online algorithm which selects a subset $\pi_\bx(R(\bx)) \subseteq R(\bx)$ such that $\1_{\pi_\bx(R(\bx))} \in \cP$.
\end{definition}
In this paper we are only considering a rank-$1$ matroid as a feasibility constraint, and thus $\cP = \set{(x_1, \ldots,x_n) \geq 0 \midd \sum_i x_i \leq 1}$. Intuitively, we say that an OCRS is $c$-selectable if and only if an active element $i \in R(\bx)$ can be included in the currently selected elements $S \subseteq R(\bx)$ and maintain feasibility with probability at least $c$.

\begin{definition}[$c$-selectability]\label{def:cSelectable}
Let $c \in [0,1]$. An OCRS for $\cP$ is $c$-selectable if and only if for any
$\bx \in \cP$, it returns a set $S$ such that
\[
\Pr\brk{i \in S \midd i \in R(\bx)} \geq c, \qquad \forall i \in [n].
\]
\end{definition}

For each $i \in [n]$, $x_i$ denotes the marginal probability that $X_i$ is active and also that $\sum_i x_i \leq 1$. A strategy will sequentially inspect each $X_i$ and select at most one. The goal is to maximize $\alpha$ so that $\Pr\brk{X_i \textnormal{ selected and } X_i \text{ is active}} \geq \alpha x_i$ for all $i \in [n]$. Such a strategy is called $\alpha$-selectable.

\begin{proof}[Proof of \Cref{thm:ocrs}]
We first show a  $\prn{\frac{1}{1+e^{4\Delta}}}$-selectable OCRS. The algorithm is quite simple: for each $i \in [n]$, if we reach $i$ and $X_i$ is active, we select $X_i$ with probability $q_i$ where
\[
q_i := \frac{\alpha x_i}{\Pr\brk{X_i \text{ is active and } X_1, \dots, X_{i-1} \text{ not selected}}}.
\]
If $q_i \leq 1$ for all $i$, then this is a valid $\alpha$-OCRS strategy as $\Pr\brk{X_i \text{ is active and selected}} = \alpha x_i$.
To show that $q_i \leq 1$, notice that
\begin{align*}
&\Pr\brk{X_i \text{ active \& } X_1, \dots, X_{i-1} \text{ not selected}} \\
&\hspace{5cm} \geq e^{-4 \Delta} \Pr\brk{X_i \text{ active}} \cdot \Pr\brk{X_1, \dots, X_{i-1} \text{ not selected}} \\
\textstyle &\hspace{5cm} \geq e^{-4 \Delta} x_i \cdot \prn{1 - \sum_{j=1}^{i-1} \alpha x_i} \\
&\hspace{5cm} \geq e^{-4 \Delta} x_i \cdot (1 - \alpha)
~=~ e^{-4 \Delta} \prn{1 - \frac{1}{1 + e^{4 \Delta}}} x_i ~=~\alpha x_i,
\end{align*}
where the first inequality follows from \Cref{lem:mrf-conditioning} and the third from $\sum_i x_i \leq 1$.

For the upper bound, let $p = \frac{1}{1 + e^{\Delta}}$, $q = \frac{1}{1 + e^{3\Delta}}$, and $n+1 = \floor{\frac{p+q}{q}} \approx e^{2\Delta}$. Our MRF will again be a path with $E = \{\{i-1, i\} \colon i \in [n]\}$, like in \Cref{fig:path-mrf}. Let $\bX \in \{0, 1\}^{n+1}$ be a sample from an MRF with
\[
\Omega_0 = \set{0,1}, \text{ and }
\Omega_i = \set{0, 1}, \: \:
\psi_{i-1, i}(x_{i-1}, x_i) = \begin{cases}
-\Delta & \text{if } x_{i-1} \neq x_i \\
\Delta & \text{if } x_{i-1} = x_i.
\end{cases}, \: \text{ for } 1 \leq i \leq n.
\]
Also, using Lemma~\ref{lem:mrf-choose-probs}, we can choose $\psi_i(x)$ for $i \in [n]$ (in descending order of $i$) so that $\Pr\brk{X_i = 1 \midd X_{i-1} = 0} = p$. Notice that this also ensures that $\Pr\brk{X_i = 0 \midd X_{i-1} = 1} = q$. This implies that $\bX$ is a Markov chain with transition matrix
$
P = \begin{bmatrix}
1-p & q \\
p & 1-q
\end{bmatrix}.
$
The stationary distribution of this chain is $\frac{1}{p+q} \begin{bmatrix} q \\ p \end{bmatrix}$. Finally, we also choose $\psi_0(x)$ such that $\Pr\brk{X_0 = 1} = \frac{q}{p+q}$, so the Markov chain begins in its stationary distribution. This ensures that $\Pr\brk{X_i = 1} = \frac{q}{p+q}$ for all $i$, so we choose $x_i = \frac{q}{p+q} = \Pr\brk{X_i = 1}$.

Now, an $\alpha$-OCRS algorithm  must, upon reaching $X_i$ with $X_i = 1$, select $X_i$ with probability
\[
\frac{q\alpha}{(p+q)\Pr\brk{X_i = 1 \text{ and } X_0, \dots, X_{i-1} \text{ not selected}}}.
\]
Moreover, notice that this probability is independent of the history of values $X_0, \dots, X_{i-1}$ due to the Markov property, i.e. the future values $X_{i+1}, \dots, X_{n}$ are independent of $X_0, \dots, X_{i-1}$ given $X_i = 1$.

To determine which values of $\alpha$ allow for such an algorithm, define for $0 \leq i \leq n$  the vector $y_i = \brk{y_i^{(0)}, y_i^{(1)}}^T$ by 
$ y_i^{(v)} := \Pr\brk{X_i = v \text{ and } X_0, \dots, X_i \text{ not selected}}.$ 
We have $y_0 = \frac{1}{p+q} \begin{bmatrix} q \\ p \end{bmatrix} - \frac{1}{p+q}\begin{bmatrix} q \alpha \\ 0 \end{bmatrix}$ and
$y_i = P y_{i-1} - \frac{1}{p+q} \begin{bmatrix} q \alpha \\ 0 \end{bmatrix}.$
By induction on $y_{i-1}$, we get
\begin{align*}
y_i &= \frac{1}{p+q} \prn{\begin{bmatrix} q \\ p \end{bmatrix} - \sum_{k = 0}^{i} P^k \begin{bmatrix} q \alpha \\ 0 \end{bmatrix}} \\
&= \frac{1}{p+q} \prn{\begin{bmatrix} q \\ p \end{bmatrix} - \sum_{k = 0}^{i} \prn{\frac{q \alpha}{p+q} \begin{bmatrix} q \\ p \end{bmatrix} + (1-p-q)^k \frac{p q \alpha}{p+q} \begin{bmatrix} 1 \\ -1 \end{bmatrix}}}.
\end{align*}
In order for the selectability of the OCRS to be $\alpha$, we must have
\begin{align*}
0 \leq y_n^{(1)} &= \frac{1}{p+q} \prn{q - \sum_{k = 0}^n \prn{\frac{q^2 \alpha}{p+q} + (1-p-q)^k \frac{p q \alpha}{p+q} }} \\
&=\frac{q}{p+q} \prn{1 - \frac{q n \alpha}{p+q} - \frac{1 - (1-p-q)^{n+1}}{p+q} \cdot \frac{p \alpha}{p+q}} .
\end{align*}
Using $\frac{q n}{p+q} \geq 0$, we have
\begin{align*}
\frac{1}{\alpha} &\geq \frac{1 - (1-p-q)^{n+1}}{p+q} \cdot \frac{p}{p+q} 
\geq \frac{1 - e^{-(p+q)(n+1)}}{2(p+q)} \geq \frac{1 - e^{-\f{1}{2q}}}{2(p+q)} ~\geq~ \frac{1}{4(p+q)},
\end{align*}
which implies $\alpha \leq 4(p+q) \leq 4 e^{-\Delta}$.
\end{proof}


\balance
\printbibliography
\appendix
\section{Omitted Proofs}

\subsection{Proof of Lemma~\ref{lem:gen-rev-n-srev}}\label{apx:rev-n-srev-proof} 

For each realization $\bss \in \Omega$, let $v_\bss$ denote the buyers valuation when $\bt = \bss$. Let $\xi_i$ be the set of $\bss \in \Omega$ such that $i$ is the smallest index with $v_\bss(i) = \max_j v_\bss(j)$. In other words, $\xi_i$ is the set of realizations of $\bt$ where $i$ is the ``favorite item''.
    
Let $D_{\xi_i}$ denote the distribution of $v$ given $\bt \in \xi_i$, and let $r_{\xi_i}(\bss)$ denote the revenue obtained from the optimal mechanism for $\Rev(D_{\xi_i})$ for the realization $\bt = \bss$. We have
\begin{align*}
\Rev(D) &\leq \sum_{i \in [n]} \Pr\brk{\bt \in \xi_i} \Rev\prn{D_{\xi_i}} \\
&\leq \sum_{i \in [n]} \sum_{\bss \in \xi_i} \Pr\brk{\bt = \bss} \cdot r_{\xi_i}(\bss) \\
&\leq e^{4 \Delta} \sum_{i \in [n]} \sum_{\bss \in \xi_i} \Pr\brk{\lt_i = \lss_i} \cdot \Pr\brk{\bt_{-i} = \bss_{-i}} \cdot r_{\xi_i}(\bss) \\
&= e^{4 \Delta} \sum_{i \in [n]} \Pr\brk{\bt' \in \xi_i} \cdot \Rev\prn{D'_{\xi_i}},
\end{align*}
where the third inequality follows from Lemma~\ref{lem:mrf-conditioning}, and $\bt'$ has distribution
\[
\Pr\brk{\bt' = \bss} = \Pr\brk{t_i = s_i} \cdot \Pr\brk{\bt_{-i} = \bss_{-i}},
\]
and $D'_{\xi_i}$ is the distribution of $v_{\bt'}$ given that $\bt' \in \xi_i$. Intuitively, $\bt'$ is a copy of $\bt$ where we make $t_i'$ to be independent of $\bt_{-i}'$ with the same marginal distribution.

Using \Cref{lem:gen-splitting}, we find
\[
\Pr\brk{\bt' \in \xi_i} \cdot \Rev\prn{D'_{\xi_i}} \leq 2\Pr\brk{\bt' \in \xi_i} \cdot \prn{\Val_{-i}(D'_{\xi_i}) + \Rev_i(D'_{\xi_i})},
\]
where $\Val_{-i}(D)$ denotes $\E_{v \sim D} [v([n] \setminus \{i\})]$. 

First, we have
\[
\Pr\brk{\bt' \in \xi_i} \cdot \Rev_i(D'_{\xi_i}) \leq \Rev_i(D') = \Rev_i(D).
\]
Additionally, we claim
\[
\Pr\brk{\bt' \in \xi_i} \cdot \Val_{-i}(D'_{\xi_i}) \leq \rho \cdot \Rev_i(D).
\]
To see this, consider the mechanism for selling $i$ that samples $\bt' \sim D'$ and sets a price of $\max_{j \neq i} v_{\bt'}(j)$. We can consider $v_{\bt'}(i)$ to be the buyer's valuation on $i$, since it is independent. Then the item sells with probability at least $\Pr(\bt' \in \xi_i)$ and generates revenue at least $\E\brk{\max_{j \neq i} v_{\bt'}(j) \mid \bt' \in \xi_i} \geq \frac{1}{\rho} \Val_{-i}(D'_{\xi_i})$, implying the claim.

Therefore, we have
\[
\Rev(D) \leq e^{4 \Delta} \sum_{i \in [n]} (\rho+1) \cdot \Rev_i.
\]
Putting everything together, we obtain
\[
\Rev(D) \leq 2(\rho + 1) e^{4\Delta} \cdot \sum_{i \in [n]} \Rev_i(D).   \qedhere
\]

\end{document}